\documentclass[12pt]{amsart}

\usepackage[T1]{fontenc}
\usepackage[utf8]{inputenc}
\usepackage[english]{babel}

\usepackage{amssymb}
\usepackage{array}
\usepackage{nicefrac}

\usepackage{setspace} 
\usepackage[a4paper,tmargin=1in,bmargin=1.2in ,lmargin=1in,rmargin=1in,headheight=1in,headsep=1in,footskip=1.5\baselineskip]{geometry}

\usepackage{natbib}
\usepackage{fancyhdr}
\usepackage{bbm}
\usepackage{xcolor}
\usepackage{hyperref}
\hypersetup{
  colorlinks=true,
  linkcolor=blue,
  citecolor=blue
}
\usepackage{mleftright}
\usepackage{subcaption}
\usepackage{paralist}
\usepackage{caption}
\usepackage[normalem]{ulem}

\theoremstyle{plain}
\newtheorem{theorem}{Theorem}
\newtheorem*{theorem*}{Theorem}
\newtheorem{lemma}{Lemma}
\newtheorem{proposition}{Proposition}

\newtheorem{example}{Example}

\newtheorem{assumption}{Assumption}
\newtheorem*{lemma*}{Lemma}

\theoremstyle{definition}
\newtheorem{definition}{Definition}
\newtheorem*{definition*}{Definition}

\usepackage{graphicx}
\usepackage{pstricks, enumerate, pst-node, pst-text, pst-plot}

\numberwithin{equation}{section}

\usepackage{xparse}
\DeclareDocumentCommand\Pr{ m g }{\ensuremath{
    {   \IfNoValueTF {#2}
      {\mathbb{P}\mleft[{#1}\mright]}
      {\mathbb{P}\mleft[{#1}\middle\vert{#2}\mright]}
    }
}}
\DeclareDocumentCommand\E{ m g }{\ensuremath{
    {   \IfNoValueTF {#2}
      {\mathbb{E}\mleft[{#1}\mright]}
      {\mathbb{E}\mleft[{#1}\middle\vert{#2}\mright]}
    }
}}

\usepackage{bm}
\usepackage{pifont}
\newcommand{\greentick}{{\color{green}\ding{51}}}
\newcommand{\redcross}{{\color{red}\ding{55}}}

\newcommand\restr[2]{{
		\left.\kern-\nulldelimiterspace
		#1
		\vphantom{\big|}
		\right|_{#2}
}}

\usepackage{tikz-cd}

\makeatletter
\def\th@plain{
  \thm@headfont{\bfseries}
  \itshape
  \thm@preskip=4pt
  \thm@postskip=4pt
}
\makeatother

\makeatletter
\renewenvironment{proof}[1][\proofname]{\par
  \pushQED{\qed}
  \normalfont
  \topsep=3pt\relax
  \trivlist
  \item[\hskip\labelsep\itshape #1\@addpunct{.}]\ignorespaces
}{
  \popQED\endtrivlist\@endpefalse
}
\makeatother

\makeatletter
\renewcommand{\subsection}{\@startsection{subsection}{2}{\z@}
  {-1.5ex\@plus -.5ex \@minus -.2ex}
  {3pt}
  {\normalfont\bfseries}}
\makeatother

\begin{document}

\title[]{Calibrated Forecasting and Persuasion}

\author{Atulya Jain}
\thanks{Department of Economics, University of Bonn. Email: \href{mailto:ajain@uni-bonn.de}{ajain@uni-bonn.de}}

\author{Vianney Perchet}
\thanks{ENSAE and Criteo AI Lab. Email: \href{mailto:vianney.perchet@normalesup.org}{vianney.perchet@normalesup.org}}

\thanks{Atulya Jain thanks Tristan Tomala, Nicolas Vieille, and Frédéric Koessler for their invaluable advice and encouragement throughout the project. We also thank  Andreas Kleiner, Jérôme Renault, and various seminar and conference audiences for helpful comments. The research conducted by Atulya Jain was supported by the DATAIA convergence institute (ANR-17-CONV-0003), operated by HEC Paris, and by the German Research Foundation (DFG) under Germany's Excellence Strategy -- EXC 2126/2 -- 390838866. This work benefited from Hi! PARIS and State funding managed by the French National Research Agency (ANR) under the France 2030 program, reference ANR-23-IACL-0005. An extended abstract of this paper appeared in the Proceedings of the ACM Conference on Economics and Computation (EC'24).}

\address{}
\date{\today}

\begin{abstract}
We study a dynamic game where an expert sends probabilistic forecasts to a decision-maker. The decision-maker verifies these forecasts using a calibration test based on past data. How should the expert send forecasts to maximize her payoff while passing the test? For a stationary ergodic process, we characterize the optimal forecasting strategy by reducing the dynamic game to a static persuasion problem. Calibration restricts long-run forecast distributions to mean-preserving contractions of the distribution of truthful forecasts. We compare the payoffs attainable by an informed and uninformed expert, providing a benchmark for the value of information. Finally, we consider a decision-maker who uses a mean-based learning strategy and show that the expert can guarantee at least the calibration benchmark and sometimes strictly more.
\end{abstract}

\maketitle

\section{Introduction}

Probability forecasts are widely used by experts to provide information about uncertain outcomes.  These forecasts shape the beliefs of decision-makers and guide them to take specific actions. For example, an investor relies on forecasts by a financial analyst to determine which asset will perform best. However, the decision-maker follows  the expert's forecasts only if they are credible. A standard way to determine credibility is to perform statistical tests on the realized outcomes and verify the expert's claims. We focus on an objective test:  \textit{calibration}, which is based on the frequentist interpretation of probability. The  idea is to check whether the forecast of an outcome is close to  the actual proportion of times the outcome occurred when the forecast was made. For example, the investor checks whether an asset outperformed others on $70\%$ of the days when the analyst predicted a $0.7$ chance of it being the best. Note that a single day's outcome--whether the asset actually outperformed others--cannot prove the expert right or wrong. However, over time, one can collect data on the forecasts and the realized outcomes to evaluate the expert's credibility.   Calibration is central to forecasting and is used to assess the accuracy of prediction markets \citep{PageClemenPredictionMarket}. 

The decision-maker relies on accurate forecasts to take  optimal actions. But in many  settings, the expert herself has  skin in the game  and is affected by the decision-maker's action. This misalignment of preferences leads to \textit{strategic forecasting}. For instance, if a financial analyst earns a substantial commission upon the purchase of a certain asset, her forecasts may be biased in favor of this asset.  We study the extent of an expert's payoff gain  from strategic forecasting in an infinite horizon dynamic game.  Our main focus is on an \textit{informed expert} who knows the data-generating process and can produce  calibrated forecasts by providing truthful forecasts.  Given that she is tested by a  calibration test, how should the  expert send forecasts to maximize her payoff?

To do this, we develop a dynamic sender-receiver model where the state of the world evolves according to a stochastic process. In each period, the  sender  first provides  a probabilistic forecast of that period's state to the receiver. The receiver then uses a calibration test to assess the sender's credibility. If the sender passes the test, the receiver takes the forecasts at face value and chooses a best response  to this forecast.  Otherwise, the  sender incurs a punishment cost in that period. We focus on high punishment costs, since otherwise the sender may prefer uncalibrated strategies. Finally, the realized state is observed before moving to the next period.

Our approach differs from standard Bayesian models, where the receiver is assumed to know the distribution of the process or at least to hold a prior over it.  In such models, given the sender's strategy, the receiver updates  beliefs from the observed forecasts and chooses  actions accordingly. Handling such beliefs in equilibrium\textemdash even in simple environments\textemdash is complex  and imposes strong rationality assumptions.  In contrast, we assume the receiver has  no prior information about either the process or the sender's strategy.  Instead, he  evaluates the sender's credibility using a calibration test on the history of forecasts and  states.

We first analyze the case of an informed sender who knows the  data-generating process. She can pass the calibration test by forecasting truthfully. However, other forecasting strategies may also pass the test and result in higher payoffs. We show that calibration restricts the feasible long-run distribution of forecasts. For a stationary and ergodic process, a finitely supported distribution of forecasts can arise under calibration if and only if it is a mean-preserving contraction of the distribution of truthful forecasts (Lemma~\ref{lemma:stationarympc}). Intuitively, this corresponds to merging low-value and high-value truthful forecasts with appropriate weights so that the overall forecast matches the true frequencies. Thus, while the forecasts must be correct on average, they can be less informative than truthful forecasts. This characterizes the extent of strategic forecasting permitted by calibration.
	
Our main result (Theorem \ref{maintheorem}) establishes that the informed sender's maximum payoff in the dynamic forecasting game equals the value of a  static persuasion problem. The state,  signal, and  prior in the persuasion problem correspond to the conditional probability, the forecast, and the distribution of conditionals  in the dynamic game. In Bayesian persuasion \citep{BPKamenic}, the sender commits to a signaling policy that maximizes her expected utility. We construct a forecasting strategy that implements an optimal signaling policy period by period, so forecasts depend only on the current conditional. For the upper bound, we show that any payoff gain from failing the calibration test is bounded proportionally to the frequency of failures and outweighed by sufficiently high punishment costs.  Thus, the sender's dynamic problem reduces to a static persuasion problem. An interesting implication is that calibration can substitute for commitment: the sender achieves the Bayesian persuasion payoff. \citet{GuoShmaya} obtain a related result in a static forecasting model with a miscalibration cost proportional to the expected distance between forecasts and their true conditional distributions. In our dynamic model, the receiver instead evaluates calibration using the history of forecasts and realized states, and the sender incurs a fixed cost whenever she fails the test.

We next consider an uninformed sender who does not know the data-generating process. As shown by  \cite{FosterVohra}, even an uninformed sender can  produce  calibrated forecasts for any stochastic process. We derive payoff guarantees for the uninformed sender in two environments and relate them to benchmarks from Bayesian persuasion. First, in an adversarial setting where nature can adapt its strategy, there is no fixed prior. In this case, the sender can guarantee the closed convex envelope of her approximate indirect utility, evaluated at the empirical distribution of states (Theorem~\ref{maxattainable}).\footnote{For the
uninformed sender, we restrict forecasts to a finite set and
focus on $\epsilon$-calibration. Thus, our payoff guarantees
are approximate.} This parallels the sender's worst payoff in canonical Bayesian persuasion, with the empirical distribution of states playing the role of the prior. Second, when the unknown process is an ergodic Markov chain, the uninformed sender can do better. In particular, she can guarantee the finite-grid counterpart of the no-disclosure payoff in the persuasion problem with prior given by the distribution of conditionals\textemdash the same problem as in the informed sender case (Proposition~\ref{uninformedstationary}).  This comparison separates credibility from information. An expert need not know the data-generating process to produce calibrated forecasts; information is valuable  because it allows the sender to benefit more from strategic forecasting.

We apply our model to analyze a financial platform that provides forecasts about a Markovian state. The platform's payoff depends on forecast precision and user engagement. This creates a trade-off: precise forecasts enhance reputation but reduce  the time a user spends on the platform. We derive  the platform's optimal forecasting strategy using the concavification  approach  and illustrate it graphically. We show that it may be optimal for the platform to send coarse forecasts rather than truthful ones.

One might question whether calibration is an effective decision-making heuristic, given that it permits strategic forecasting. To address this concern, we also model the receiver's behavior using regret minimization, a widely studied approach in online settings. Regret compares the receiver's realized payoff to the payoff he would have obtained by playing the best fixed action in hindsight.  Regret minimization is closely linked to calibration (see \citealp{vianney}), and we use this connection to compare strategic forecasting under the two decision rules. We show that if the receiver  myopically best responds to the forecasts of any calibrated  strategy, then he incurs no regret, thereby justifying the use of calibration  as a decision-making heuristic (Proposition~\ref{calibnimpliesnoreg}). We next focus on a natural class of learning strategies, called mean-based algorithms. These algorithms assign vanishing probability to actions whose cumulative payoff is substantially lower than that of another action. Against a mean-based receiver, the informed sender can guarantee the calibration benchmark, namely the persuasion value, and sometimes strictly more (Proposition~\ref{noregretQ} and Theorem~\ref{regretmuchmore}).

The rest of the paper is organized as follows.	In Section \ref{sectionmodel}, we introduce the  dynamic forecasting game and  the calibration test. In Section \ref{sectionMainresults}, we present the persuasion problem, state our main results and  provide an application.   In Section \ref{sectionnoregret}, we consider a receiver who minimizes regret instead of using the calibration test.  Section~\ref{sectionrelated} discusses the related literature. In  Section \ref{sectionconclusion}, we conclude and discuss  future work. All omitted proofs are in Appendix \ref{sec:proofs}. 

\section{Dynamic forecasting game} \label{sectionmodel}
We consider a dynamic  game between a sender (she) and a receiver (he), where the state of the world evolves over time. In each period,  the sender sends a forecast about that period's unknown state. The receiver observes the forecast and then chooses an action. The realized state  is observed before  proceeding to the next period. 

Let $\Omega$ denote the finite set of states, $F \subseteq \Delta (\Omega)$ denote the set of  feasible forecasts over the states, and $A$ denote the finite set of actions.\footnote{$\Delta(X)$ denotes the set of all probability distributions over the set $X$.} Unless specified otherwise, all forecasts are feasible, i.e.,  $F= \Delta (\Omega)$. The state ${(\omega_t)}_{t \geq 1}$ evolves over time and is governed by a stochastic process with distribution $\mu \in \Delta (\Omega^\infty)$.  Denote by $\omega_t$  the state in period $t$, and by  $\omega^t=(\omega_1,...,\omega_{t-1})$ the  history of  states up to period $t$.   We assume the sender is informed and  knows the  distribution $\mu$, whereas the  receiver is uninformed.\footnote{ Section \ref{secuninformed} considers the case where the sender is uninformed.} In particular, given the history $\omega^t$, the sender can compute the  \textit{conditional probability} in period $t$, denoted by  $p_t = \mu(\cdot \mid \omega^{t}) \in \Delta (\Omega)$.  Hence, the sender knows the objective probability over states   at each period. We assume that the  conditionals  take values from a finite set $D \subset \Delta (\Omega)$.\footnote{This is a restriction on the stochastic process $\mu$, rather than on the set of feasible forecasts. It holds, for example, for finite-state Markov chains. The assumption is imposed for tractability; more general processes could be treated using finite approximations of the set of conditionals.}

In each period, the sender's forecast  may depend on the entire history of past forecasts and  states. Formally, a forecasting strategy is a map $\sigma: \bigcup_{t \geq 1} (F \times \Omega)^{t-1} \to \Delta (F)$. Let $f_t$ denote the realized forecast in period $t$. After observing $f_t$, the receiver chooses an action $a_t \in A$ and then the state $\omega_t \in \Omega$ is realized. The sender's and receiver's payoffs depend only on the realized state and action and are given by $u_S(\omega_t,a_t)$ and $u_R(\omega_t,a_t)$, respectively. A forecasting strategy is \textit{truthful} if $f_t = p_t$ in every period.\\

\textbf{Calibration Test:} To model the receiver's behavior, we adopt a frequentist approach based on calibration. The receiver has no prior information about the stochastic process or the sender's strategy. Instead, he  evaluates the sender's credibility by  comparing her forecasts with the realized frequencies of states. In each period, he computes the calibration error using  the history of forecasts and states. The sequence of error margins $(\epsilon_t)_{t\geq1}$ is fixed ex ante, with $\epsilon_t$ specifying the tolerance for calibration error over the first $t$ periods.

\begin{definition} \label{calibrationdef}
	A $T$-sequence of forecasts $(f_t)_{t=1}^{T}$ is  $\epsilon_T$-\textit{calibrated}  if 
	\begin{equation}
		\sum_{f \in F}\frac{N_{T}[f]}{T}  \| \overline{\omega}_T[f] - f \| \leq \epsilon_T,
	\end{equation}
	where $ N_{T}[f]$ and  $\overline{\omega}_T[f]$ denote, respectively, the number of periods  in which the forecast equals $f$ up to period  $T$ and  the empirical distribution of states in those periods:
	\begin{equation*}
			N_{T}[f] := \sum_{t=1}^T \mathbf{1}_{\{f_t=f\}}, \qquad  \overline{\omega}_T[f]:= \frac{1}{N_T[f]} \sum_{t=1}^T \mathbf{1}_{\{f_t=f\}} \delta_{\omega_t},
	\end{equation*} 
	where $\delta_\omega \in \Delta (\Omega)$ is the Dirac distribution on state $\omega$ and  $\| \cdot \|$ is  the  Euclidean norm.\footnote{For completeness, if $N_T[f] =0$, define $\overline{\omega}_T[f] = f$.}
\end{definition}

A $T$-sequence of forecasts is $\epsilon_T$-calibrated if the frequency-weighted average distance between forecasts and the empirical distributions of states  does not exceed the error margin $\epsilon_T$.\footnote{This holds even if the set  $F$ is infinite, as the sum in the definition is implicitly restricted to the forecasts actually sent, which are finite. } We illustrate the calibration test using a  rain forecasting example.

\begin{example} \label{binaryexample}
Consider the states $\Omega=\{0,1\}$, where $\omega=1$ denotes rain and $\omega=0$ denotes no rain.   The weather evolves according to a Markov chain: with probability $0.8$ the state remains the same from one period to the next, that is,
$\mu(\omega_{t+1}=1 \mid \omega_t=1)=\mu(\omega_{t+1}=0 \mid \omega_t=0)=0.8$ for every $t \in \mathbb{N}$. In each period, the sender predicts the probability of rain.  Consider three sequences of forecasts, denoted $F1$, $F2$, and $F3$ (Table \ref{tablerain}). $F1$ corresponds to truthful forecasts, where the sender announces $80\%$ after rain and $20\%$ after no rain.  $F2$ corresponds to coarse forecasts obtained by garbling the truthful forecasts. $F3$ corresponds to extreme forecasts, predicting $100\%$ after rain and $0\%$ after no rain.

\begin{table}[ht]
	\centering
\caption{Sequences of forecasts and the calibration test.}
	\label{tablerain}
	\begin{tabular}{l||ccccccccccc}
		\textbf{Period} & 1 & 2 & 3 & 4 & 5 & 6 & 7 & 8 & 9 & 10 & \\ \hline
		\textbf{State} & 1 & 0 & 0 & 0 & 0 & 0 & 1 & 1 & 1 & 1 & \\ \hline
		\textbf{F1} & 80\% & 80\% & 20\% & 20\% & 20\% & 20\% & 20\% & 80\% & 80\% & 80\% & \greentick \\
		\textbf{F2} & 60\% & 60\% & 40\% & 40\% & 60\% & 40\% & 40\% & 60\% & 60\% & 40\% & \greentick \\
		\textbf{F3} & 100\% & 100\% & 0\% & 0\% & 0\% & 0\% & 0\% & 100\% & 100\% & 100\% & \redcross \\
		\hline
	\end{tabular}
\end{table}

For the error margin $\epsilon_{10}=0.1$, both $F1$ and $F2$ pass the $\epsilon_{10}$-calibration test, while $F3$ does not. For instance, when $F2$ predicts a $60\%$ chance of rain, it rains in three out of five periods, and similarly for $40\%$. In contrast, when $F3$ predicts a $100\%$ chance of rain, it does not always rain. Thus, both the truthful forecasts $F1$ and the coarse forecasts $F2$ pass the test. Intuitively, in the long run, calibration preserves the correct mean while allowing forecasts to be less informative than the truth.
\end{example}

\textbf{Pass}: The sender passes the calibration test in period $t$ if the sequence of past forecasts $(f_i)_{i=1}^{t-1}$ is $\epsilon_{t-1}$-calibrated.\footnote{In period $t=1$, there is no past data, so the sender passes the test trivially.} In this case, the receiver takes the forecast at face value and chooses $a_t=\hat{a}(f_t)$, where $\hat{a}(f)$ denotes his best response to belief $f$, with ties broken in favor of the sender.  Formally,
\begin{equation*}
	\hat{a}(f) \in \arg \max_{a \in A} \sum_{\omega \in \Omega} f(\omega)u_R(\omega, a).
\end{equation*}

\textbf{Fail}: If the sequence of past forecasts is not $\epsilon_{t-1}$-calibrated, the sender fails the calibration test in period $t$. In this case, her payoff in that period is $-c$, where $c>0$ is the punishment cost. The receiver refrains from playing according to the sender's forecast until she passes the test in a future period. This punishment can be interpreted as the receiver reverting to a default action.\footnote{This can be modeled with an additional default action $a_d$, used when the sender fails the calibration test, such that $u_S(\omega,a_d)=-c$ and $u_R(\omega,a_d)=0$ for all $\omega\in\Omega$.
In financial forecasting, for example, this may correspond to the investor not buying any asset, resulting in zero commission for the analyst.} Alternatively, it may reflect a loss of credibility, with an intermediary forwarding the sender's forecasts only when she passes the calibration test.

The sender's goal is to choose a  forecasting strategy $\sigma^\star$ that maximizes her expected long-run average  payoff:
\begin{equation} \label{constrainedmax}
	\liminf_{T \to \infty} \frac{1}{T} \sum_{t=1}^T \mathbb{E}_{\sigma^\star,\mu} \left[ u_S(\omega_t, a_t) \right].
\end{equation}

We refer to the strategy $\sigma^\star$ as the \textit{optimal forecasting strategy}.  To make the problem tractable and economically meaningful, we impose two assumptions.

First, observe that for a sufficiently small error margin $\epsilon_T$, even a truthful sender may fail the $\epsilon_T$-calibration test with positive probability. For instance, in Example~\ref{binaryexample}, when $\epsilon_1=0.1$, the truthful forecasts are not $\epsilon_1$-calibrated regardless of which state $\omega_1$ is realized. Nevertheless, as more data are observed, empirical frequencies converge to the truthful forecasts. This motivates an asymptotic notion of calibration.
	
	\begin{definition} \label{calibasymdef}
		A forecasting strategy $\sigma $  is $\epsilon$-\textit{calibrated}  if
		\begin{equation} \label{eq:calibrationasym}
			\limsup_{T \to \infty}  \sum_{f \in F} \frac{N_{T}[f]}{T}  \| \overline{\omega}_T[f] -  f\|   \leq \epsilon, \quad \mathbb{P}_{\sigma,\mu}\text{-a.s.} 
		\end{equation}
		A forecasting strategy $\sigma$ is \textit{calibrated} if it is $\epsilon$-calibrated, for every $\epsilon > 0$. 
	\end{definition}

Under calibration, the empirical distribution of states conditional on any forecast used with positive long-run frequency converges to that forecast. The truthful forecasting strategy is calibrated for any stochastic process \citep{dawid1982well}.

A natural requirement is that the error margins in the finite calibration test converge to zero over time. However, if they shrink too quickly, even truthful forecasting may fail the test too often. We therefore impose the following requirement.

\begin{assumption}\label{assumerror}
Fix a stationary ergodic process $\mu$. The error margins satisfy $\epsilon_t\to0$ and are such  that, for every finite $F_0\subseteq F$, any calibrated strategy induced by a fixed rule $\pi:D\to\Delta(F_0)$ fails the calibration test in a vanishing fraction of periods in expectation.
\end{assumption}

Assumption~\ref{assumerror} ensures that  calibrated strategies induced by a fixed rule, including truthful forecasting, are not rejected too often. Proposition~\ref{prop:slowerror} in Appendix~\ref{sec:proofs} shows that such a sequence exists for every stationary ergodic process; in fact, for this sequence, every such strategy then fails the test only finitely many times almost surely.\footnote{Uniform rates of convergence need not hold across arbitrary strategies and processes unless stronger dependence conditions, such as mixing or irreducibility in Markov chains, are imposed.}

Second, we use calibration as the credibility criterion for the sender. Although an informed sender can always be calibrated by forecasting truthfully, she may prefer an uncalibrated strategy when the punishment cost is low. We therefore focus on sufficiently high costs.

\begin{assumption}\label{assum:highcost}
The punishment cost satisfies $c > (1+\sqrt{2})L$, where
\[
L:=\max_{a \in A} \| u_S(\cdot,a) \|.
\]
\end{assumption}

As we show later, sufficiently high punishment costs imply that any optimal
forecasting strategy fails the calibration test in only a vanishing fraction of
periods in expectation. This implication is not immediate: a sender may benefit from occasional failures if they make the distribution of states in the periods in which she passes more favorable. We show that any such gain is bounded proportionally to the frequency of failures.

	\section{Main results} \label{sectionMainresults}
	
Our main result characterizes the informed sender's optimal payoff for a stationary ergodic process by reducing the dynamic forecasting game to a static persuasion problem. Before stating the main result, we introduce the persuasion problem. We then analyze the payoff that an uninformed sender can guarantee. Finally, we illustrate the results with an application to a financial forecasting platform.

\subsection{ Persuasion  problem} \label{subsectionpersuasiongame}

We consider a static persuasion problem whose parameters are directly linked to those of the dynamic forecasting game. The state space $\Delta(\Omega)$ corresponds to conditional probabilities, and the prior $P \in \Delta(\Delta(\Omega))$ is given by their distribution. The sender commits to a signaling policy $\pi:\Delta(\Omega)\to\Delta(\Delta(\Omega))$, where signals correspond to forecasts.\footnote{To clarify, the sender does not have commitment power in the dynamic forecasting game.} Once the conditional $p\in\Delta(\Omega)$ is realized, the sender sends a forecast $f\in\Delta(\Omega)$ according to $\pi(\cdot\mid p)$. Each forecast induces a posterior belief over conditionals and, consequently, a posterior mean. Thus, a signaling policy induces a distribution of posterior means $Q\in\Delta(\Delta(\Omega))$. Without loss of generality, we restrict attention to policies under which the posterior mean induced by forecast $f$ equals $f$. It is well known that a distribution $Q$ is implementable by a signaling policy if and only if it is a mean-preserving contraction of the prior $P$ (e.g., \citealp{gentzkow2016rothschild,kolotilin2018optimal,kleiner2021extreme,arieli2023optimal}).

A probability distribution with finite support can be represented as 
$P=(\bm{\lambda},\bm{p})$, where 
$\bm{p}=(p_1,\ldots,p_n)$, with $p_i \in \Delta(\Omega)$, denotes the support and $\bm{\lambda}=(\lambda_1,\ldots,\lambda_n)$ the associated weights. Similarly, write $Q=(\bm{\nu},\bm{f})$, where $\bm{f}=(f_1,\ldots,f_m)$ denotes the support and $\bm{\nu}=(\nu_1,\ldots,\nu_m)$ the associated weights. We follow \citet{EltonHillFusion} and \citet{WhitmeyerJoseph2019MoMC} in defining mean-preserving contractions.

\begin{definition}\label{defsmpc}
	A  probability distribution $Q=(\bm{\nu},\bm{f})$ is a \textit{simple mean-preserving contraction}   of $P=(\bm{\lambda},\bm{p})$, if there exists a row-stochastic matrix $G\in \mathbb{R}^{n \times m}$ satisfying
	\begin{align} \label{smpc}
		& \bm{\lambda} G=\bm{\nu},\\
		& (\bm{\lambda p})G  = (\bm{\nu f}), \label{smpc2}
	\end{align}
	where $\bm{\lambda p}=(\lambda_1 p_1,\ldots,\lambda_n p_n)$ and $\bm{\nu f}=(\nu_1 f_1,\ldots,\nu_m f_m)$.
\end{definition}

Intuitively, a simple mean-preserving contraction takes a fraction $G_{ij}$ of the mass $\lambda_i$ at each $p_i$ and pools these to form mass $\nu_j$ at $f_j$.

\begin{definition} \label{defmpc}
	$Q$ is a  \textit{mean-preserving contraction}  of $P$ if there exists a sequence of  simple mean-preserving contractions $( Q_m)_{m=1}^\infty $ such that
	 $Q_m \to Q$ (weak convergence). Denote by $\mathcal{M}(P)$ the set of all mean-preserving contractions of $P$.  
\end{definition}

			The sender's indirect  utility  from forecast $f$ is given by
			\begin{equation*}
				\hat{u}_S(f) := \sum_{\omega \in \Omega} f(\omega) \, u_S(\omega, \hat{a}(f)),
			\end{equation*}
			where $\hat{a}(f)$ is the receiver's best response to  belief $f$.

		Given prior $P$, the sender chooses  $Q \in \mathcal{M}(P)$ to maximize her expected indirect utility. Define the \textit{persuasion value} by
		\begin{equation} \label{persuasiongarble}
			\mathrm{Per}(P)
			:= \max_{Q \in \mathcal{M}(P)} \mathbb{E}_{Q}[\hat u_S].
		\end{equation}
		Let $Q^\star$ denote an optimal distribution of forecasts solving \eqref{persuasiongarble}, and let $\pi^\star$ denote a signaling policy that implements it.

When $\mathrm{Supp}(P)$ is affinely independent, the persuasion problem can be solved using the concavification approach \citep{aumann1995repeated,BPKamenic} on the restricted domain $C=\mathrm{Co}(\mathrm{Supp}(P))$.\footnote{$\mathrm{Co}(A)$ denotes the convex hull of $A$.} In particular,
\begin{equation} \label{eq:bayeaffine}
	\mathrm{Per}(P)
	=
	\mathrm{Cav}\,\restr{\hat{u}_S}{C}(m(P)),
\end{equation}
where $m(P)=\sum_{i=1}^n \lambda_i p_i$ denotes the mean of $P$ and $\mathrm{Cav}\,\restr{\hat{u}_S}{C}$ denotes the concave envelope of $\hat{u}_S$ on $C$.

This follows from a simple characterization of  $\mathcal{M}(P)$. Under affine independence, feasibility reduces to Bayes plausibility on $C$: $Q$ is feasible if and only if $\mathrm{Supp}(Q)\subseteq C$ and $m(Q)=m(P)$ \citep{wu2023geometric}. The canonical Bayesian persuasion setting is the special case in which the sender is perfectly informed, so $\mathrm{Supp}(P)=\{\delta_\omega:\omega\in\Omega\}$.

By Bauer's maximum principle the maximum in \eqref{persuasiongarble} is attained at an extreme point of $\mathcal{M}(P)$. Moreover, when $\mathrm{Supp}(P)$ is finite,  \citet{WhitmeyerJoseph2019MoMC} show that any extreme point $Q$ satisfies $|\mathrm{Supp}(Q)| \le |\mathrm{Supp}(P)|$.\footnote{For non-atomic priors on an interval, \citet{DworczakPiotr}, \citet{kleiner2021extreme}, and \citet{arieli2023optimal} show that the extreme points correspond to bi-pooling signaling policies.} Thus, without loss, an optimal distribution $Q^\star$ can be chosen to have finite support. 
	
	\subsection{Informed sender} \label{subsectiondynamic}

	In this section, we examine the dynamic game with an informed sender\textemdash one who knows the data-generating process. We show that when the process is stationary and ergodic, calibrated strategies induce  mean-preserving contractions of the distribution of conditionals. This  reduces the sender's dynamic problem to a static persuasion problem.

	Given a stochastic process $\mu$ and a forecasting strategy $\sigma$, define the
\textit{distribution of conditionals} and the \textit{distribution of forecasts},
whenever the respective limits exist, by
\begin{equation}\label{distributionofconditionals}
    C_\mu
    := \lim_{T\to\infty}\frac{1}{T}\sum_{t=1}^T\delta_{p_t},
    \qquad
    F_{\mu,\sigma}
    := \lim_{T\to\infty}\frac{1}{T}\sum_{t=1}^T\delta_{f_t}.
\end{equation}
	where  $p_t=\mu(\cdot\mid\omega_1,\ldots,\omega_{t-1})$, and $f_t$ is drawn according
to $\sigma(f_1,\omega_1,\ldots,f_{t-1},\omega_{t-1})$.
Let $\mathcal F_{\mu,\sigma}$ denote the set of all weak limit points of the
empirical distributions of forecasts. Under truthful forecasting, the empirical
distributions of forecasts and conditionals coincide. Both $C_\mu$ and $F_{\mu,\sigma}$ may depend on the realized history and are therefore random.

    For any stochastic process $\mu$, the truthful forecasting strategy is always calibrated. However, without further restrictions on the process, characterizing the set of feasible long-run outcomes that pass the calibration test is difficult since $C_\mu$ may fail to exist. We therefore restrict attention to stationary ergodic processes, for which $C_\mu$ exists and is constant almost surely (see Lemma~\ref{stationaryergodic} in Appendix~\ref{sec:proofs}). For illustration, consider Example~\ref{binaryexample}. The Markov chain is irreducible and aperiodic, with invariant distribution $(0.5,0.5)$.  Thus, rain and no rain each occur with long-run frequency $0.5$, and the corresponding conditionals are $0.8$ and $0.2$, respectively. Hence, $C_\mu$ places equal mass on $0.8$ and $0.2$. This distribution serves as the prior in the associated persuasion problem.

		We now state our main result. The maximum payoff of the informed sender in the dynamic forecasting game equals the persuasion value with prior $C_\mu$.

			\begin{theorem} \label{maintheorem}
			For a stationary ergodic process $\mu$, the informed sender's maximum payoff equals $\mathrm{Per}(C_\mu)$.
			\end{theorem}

    The key step in proving Theorem~\ref{maintheorem} is to characterize the distributions of forecasts that can arise under calibration.

\begin{lemma} \label{lemma:stationarympc}
Fix a stationary ergodic process $\mu$. If a forecasting strategy $\sigma$ is calibrated, then
$\mathcal F_{\mu,\sigma} \subseteq \mathcal{M}(C_\mu)$.
Conversely, under Assumption~\ref{assumerror}, for any
$Q \in \mathcal{M}(C_\mu)$ with finite support, there exists a calibrated
strategy $\sigma$ such that $F_{\mu,\sigma}=Q$ and that fails the calibration
test with expected long-run frequency zero.
\end{lemma}

Intuitively, calibration requires the sender to use appropriate weights to combine conditionals into forecasts whose long-run frequencies of states match the announced forecasts. These weights correspond precisely to a mean-preserving contraction of the distribution of conditionals.

\begin{proof}[Proof sketch of Theorem~\ref{maintheorem}]
The proof has two parts. We first show that no forecasting strategy can give the
sender more than the persuasion value, and then construct a strategy that attains it.

For the upper bound, suppose first that the sender always passes the calibration
test. The receiver then best responds to each announced forecast $f$.
Since the error margins converge to zero, the calibration error vanishes
in the long run. Thus, the sender's average payoff approaches the frequency-weighted average of her indirect utility $\hat u_S(f)$. By Lemma~\ref{lemma:stationarympc}, every limit point of the empirical
distribution of forecasts is a mean-preserving contraction of $C_\mu$.
Therefore, the sender's payoff cannot exceed the persuasion value
$\operatorname{Per}(C_\mu)$.

It remains to consider whether the sender can do better by occasionally failing
the calibration test. Such failures may be beneficial because they can change the
distribution of states in the periods in which the sender passes, making those
periods more favorable to her. The key step is to show that any such gain is bounded proportionally to the frequency of failures. Since each failure incurs the punishment cost \(c\), Assumption~\ref{assum:highcost} ensures that these costs more than offset the potential gain. Hence, allowing the sender to fail the test cannot raise her long-run payoff above the persuasion value. Moreover, any optimal forecasting strategy must fail the calibration test in a vanishing fraction of periods in expectation.

To attain the bound, let $Q^\star\in\mathcal M(C_\mu)$ be a finitely supported
distribution that attains the persuasion value, and let $\pi^\star$ be a signaling
policy that implements it. Stationarity and ergodicity imply that the conditionals
occur with long-run distribution $C_\mu$. In every period, after observing the
current conditional $p_t$, the sender draws her forecast according to
$\pi^\star(\cdot\mid p_t)$. Thus, in the long run, the dynamic forecasting strategy
reproduces the distribution of conditionals and forecasts generated by
$\pi^\star$ in the static persuasion problem, and hence induces $Q^\star$.
By Lemma~\ref{lemma:stationarympc} and Assumption~\ref{assumerror}, this
strategy is calibrated and fails the calibration test in a vanishing fraction of
periods in expectation. Its expected long-run payoff is therefore
$\operatorname{Per}(C_\mu)$.
\end{proof}

More generally, the restriction imposed by calibration extends beyond stationary ergodic processes. For an arbitrary stochastic process, the empirical distributions of forecasts and conditionals need not converge, so their long-run distributions may fail to exist. Nevertheless, when the sender uses finitely many forecasts, calibration still imposes an analogous restriction directly on their empirical joint frequencies; see Proposition~\ref{proposition:mpc} in Appendix~\ref{sec:proofs}.

\subsection{Uninformed sender} \label{secuninformed}
 In this section, we examine the dynamic game with an uninformed sender. A seminal result of \cite{FosterVohra} shows that even an uninformed sender can construct a calibrated forecasting strategy for any stochastic process. We derive the payoff she can guarantee and compare it with that of an informed sender. We first consider an adversarial environment in which nature chooses the sequence of states to minimize the sender's payoff, and then analyze the  case where the process is a fixed Markov chain that remains unknown to the sender.

For this section, we focus on a less ambitious goal: $\epsilon$-calibration.\footnote{This weaker criterion is common in the literature. In fact, using $\epsilon$-calibrated strategies and the "doubling trick" one can obtain calibrated strategies \citep{cesalugosi, mannorstoltz2010geometric, vianney}.} This relaxation allows us to restrict attention to a finite set of forecasts. Moreover, rather than applying the finite calibration test in every period, we impose $\epsilon$-calibration as an asymptotic requirement and let the receiver play $a_t=\hat a(f_t)$ in every period.\footnote{In the informed sender case, the error margins are tailored to a fixed stationary ergodic process. Here the process is unknown and, in the adversarial case, nature chooses states strategically, so we focus on asymptotic $\epsilon$-calibration.}

Given $\epsilon>0$, let $F_\epsilon\subseteq\Delta(\Omega)$ be a finite
$\epsilon$-grid.\footnote{Such finite discretizations of $\Delta(\Omega)$ are
standard in the calibration literature (see, e.g.,
\citealp{FosterVohra,vianney,foster2021forecast,foster2023calibeating,hart2025calibrated}).}
For each $p\in\Delta(\Omega)$, let
\[
f^*(p)\in\arg\min_{f\in F_\epsilon}\|p-f\|.
\]
Thus, $f^*(p)$ is a forecast in the grid closest to $p$, and satisfies
$\|p-f^*(p)\|\leq\epsilon$. If there are several such forecasts, choose one that
minimizes the sender's expected payoff at $p$.

 For each $p\in\Delta(\Omega)$, define the  sender's approximate indirect utility
\[
\hat u_S^\epsilon(p):=\sum_{\omega\in\Omega}p(\omega)
u_S(\omega,\hat a(f^*(p))).
\]
Thus, $\hat u_S^\epsilon(p)$ is the sender's expected payoff when the state
distribution is $p$ and the receiver best responds to the forecast $f^*(p)$. The grid approximation therefore implies that $|\hat u_S^\epsilon(p)-\hat u_S(f^*(p))|\leq L\epsilon$.
 
We begin with the adversarial case in which nature chooses the sequence of states to minimize the sender's payoff. In this environment, nature may choose the state, possibly at random,  in each period based on the entire history of past forecasts and states.

Since we do not assume that states are drawn according to any fixed stochastic process, we adopt a pathwise performance criterion and evaluate the sender's payoff against a benchmark determined by the realized empirical distribution of states. Let $\overline{\omega}_T \in \Delta (\Omega)$ denote the \textit{empirical distribution} of states up to period $T$:
\begin{equation*}
	\overline{\omega}_T:= \frac{1}{T} \sum_{i=1}^T \delta_{\omega_i}.  
\end{equation*}

The following notion of attainability formalizes this pathwise benchmark and is adapted from the online learning literature (see, e.g., \citealp{mannorTsitsiklisYu, bernsteinmanorshimkin}).
\begin{definition} \label{def:attainable}
	A function $h\colon \Delta (\Omega) \to \mathbb{R}$ is \emph{attainable} for the uninformed sender if there exists a forecasting strategy $\sigma$ such that, for any nature's strategy $\tau$,  the following conditions hold $\mathbb{P}_{\sigma,\tau}$-almost surely:
	\begin{enumerate}
		\item
		${\displaystyle
			\liminf_{T \to \infty} \left( \frac{1}{T} \sum_{t=1}^T u_S(\omega_t,a_t) - h(\overline{\omega}_T) \right) \geq 0 
		}$,
		\item
		${\displaystyle
			\limsup_{T \to \infty} \sum_{f \in F_\epsilon} \frac{N_T[f]}{T} \left\| \overline{\omega}_T[f] - f \right\| \leq \epsilon
		}$.
	\end{enumerate}
\end{definition}

The first condition requires the sender's average payoff to asymptotically exceed $h(\overline{\omega}_T)$. The second condition
requires that the forecasting strategy be $\epsilon$-calibrated.
	
We show that the closed convex hull of the  sender's approximate indirect utility is attainable. Let $\overline{\mathrm{Co}}\, g$ denote the closed convex hull of a function $g$.\footnote{Given a function $g:X \to \mathbb{R}$, over a convex domain $X$, its closed convex hull is the function whose epigraph is  $   \overline{\mathrm{Co}}   ( \{(x,r) : r \geq g(x) \})$ where $ \overline{\mathrm{Co}}  (X)$ is the closed convex hull of the set $X$. }

\begin{theorem}\label{maxattainable}
The function $\overline{\mathrm{Co}}\,\hat u_S^\epsilon$ is attainable for the uninformed sender.
\end{theorem}

This bound is the finite-grid analogue of the sender's worst payoff in the canonical Bayesian persuasion problem as a function of the prior.

\begin{proof}[Proof sketch of Theorem~\ref{maxattainable}]
The result follows from Blackwell approachability. We define a vector-valued payoff that tracks the sender's payoff, the calibration error, and the empirical distribution of states. We then define a closed and convex target set requiring the sender to obtain at least $\overline{\mathrm{Co}}\,\hat u_S^\epsilon(p)$, where $p$ denotes the empirical distribution of states, while remaining $\epsilon$-calibrated.

Blackwell's dual condition requires that, for every $p\in\Delta(\Omega)$, there exists a (randomized) forecast $x\in\Delta(F_\epsilon)$ whose expected vector payoff belongs to the target set. To verify this condition, fix any $p$ and choose the pure forecast $f^*(p)$. The calibration error is at most $\epsilon$ by construction of the grid, and the sender's expected payoff is $\hat u_S^\epsilon(p) \geq \overline{\mathrm{Co}}\,\hat u_S^\epsilon(p)$. Hence, Blackwell's condition holds and the target set is approachable. Moreover, the corresponding approachability strategy provides an explicit forecasting strategy for the uninformed sender.
\end{proof}

We now consider a non-adversarial environment, where states are governed by an unknown ergodic Markov chain of order $k$ (i.e., the next state depends only on the last $k$ states). In this case, the uninformed sender can guarantee a higher payoff while remaining $\epsilon$-calibrated.

\begin{proposition}\label{uninformedstationary}
Suppose $\mu$ is an irreducible and aperiodic Markov chain
of order $k$. Then the uninformed sender can guarantee a payoff of at least $\hat u_S^\epsilon(m(C_\mu))$ using an $\epsilon$-calibrated forecasting strategy.
\end{proposition}

In the adversarial case, nature can vary the empirical distribution of states across blocks while preserving the same overall distribution. This can reduce the payoff the sender can guarantee while remaining calibrated. In contrast, for an ergodic Markov chain, the empirical distributions over long blocks approach the same invariant distribution, ruling out such strategic variation.  This allows the sender to guarantee $\hat u_S^\epsilon(m(C_\mu))$, the finite-grid analogue of the
no disclosure payoff in the persuasion problem with prior $C_\mu$. The sender need not know in advance that the process is a Markov chain. If nature's play is favorable, she attains the higher benchmark; otherwise she secures the adversarial guarantee.\footnote{\cite{bernsteinmanorshimkin} establish this property in more general approachability settings.}

In summary, an uninformed sender can attain the function $\overline{\mathrm{Co}}\,\hat u_S^\epsilon$ in the adversarial case and guarantee $\hat u_S^\epsilon(m(C_\mu))$ when the
process is an ergodic Markov chain. Together, Theorem~\ref{maintheorem} and Proposition~\ref{uninformedstationary} provide a benchmark for the value of information in Markov
environments.

\subsection{Application: Financial Platform} \label{application}

We consider a financial forecasting platform that provides probability forecasts of a binary market state to a user. The platform's credibility is evaluated through calibration, and its payoff from a forecast depends on its precision and user engagement.

A financial platform (sender) provides daily probability forecasts about a binary market state $\omega_t \in \{H,L\}$. After observing the forecast, a user (receiver) decides whether to invest in the market. The user prefers to invest if and only if the state is $H$.

 The state evolves according to a Markov chain with transition probabilities
$\mu(\omega_{t+1}=H \mid \omega_t=H)= \mu(\omega_{t+1}=L \mid \omega_t=L)=0.95$. The initial state is drawn from the invariant distribution, which assigns equal probability to $L$ and $H$. Thus, the process is stationary and ergodic and induces the distribution of conditionals $C_\mu$ that places equal mass on $\{0.05,0.95\}$.

To illustrate the calibration restriction, we allow the platform's payoff to depend directly on its forecast. This reduced-form utility has two components. First, the platform benefits from user engagement, interpreted as time spent on the platform. Engagement is highest for  coarse forecasts and declines as forecasts become extreme. The platform's  utility   from   user engagement following forecast   $f$ (see Fig. \ref{time})  is given by: 
\begin{equation}
\hat u_S^{\mathrm{eng}}(f)=
\begin{cases}
	C(G_f) & \text{if } 0.05 \le f \le 0.95,\\
	0 & \text{otherwise},
\end{cases}
\end{equation}
where $G_f$ is a  distribution of beliefs with support $\{0.05,0.95\}$ and mean $f$,
$C(G_f)=\kappa_1\!\left(\int_0^1 \ell(q)\, dG_f(q)-\ell(f)\right)$ with $\kappa_1>0$,
and $\ell(q)=q\log(\frac{q}{1-q})+(1-q)\log(\frac{1-q}{q}).$\footnote{The user engagement utility can be microfounded via an information acquisition problem as in \cite{morrisstrack}. Starting from a prior belief $f$, the user acquires information until the belief crosses a $95\%$ confidence threshold. Then $G_f$ denotes the resulting distribution of posteriors, and the associated cost $C(G_f)$ is proportional to the expected change in the log-likelihood ratio.} Here, $C(G_f)$ measures the additional information the user acquires after observing forecast $f$: coarser forecasts induce more information acquisition and hence greater engagement.

Second, the platform derives utility from announcing precise forecasts, which we interpret as gain in reputation. The reputational utility from announcing forecast $f$ (see Fig. \ref{rep}) is given by
\begin{equation}
\hat u_S^{\mathrm{rep}}(f)=\kappa_2 (f-0.5)^2,
\end{equation}
for some constant $\kappa_2>0$.

\begin{figure}[h!]
\centering
\begin{minipage}[b]{0.32\textwidth}
	\centering
	\includegraphics[scale=0.8, width=\textwidth]{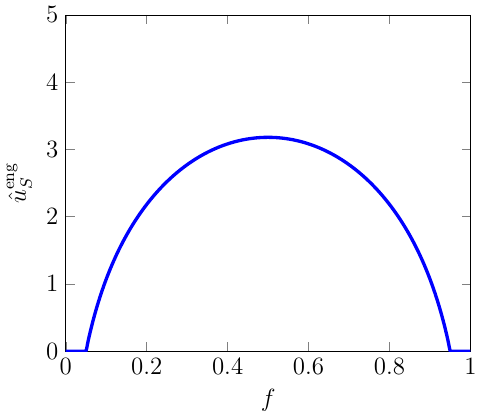}
	\caption{$\hat{u}_S^{\mathrm{eng}}(f)$}
	\label{time}
\end{minipage}
\hfill
\begin{minipage}[b]{0.32\textwidth}
	\centering
	\includegraphics[scale=0.8, width=\textwidth]{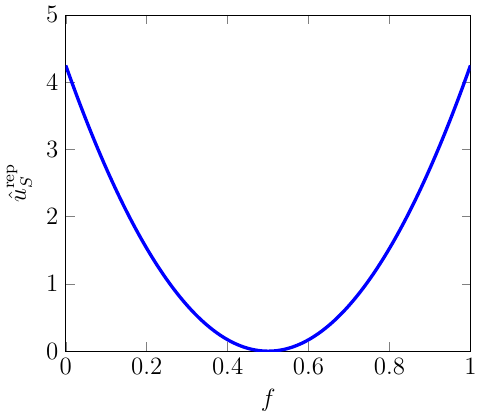}
	\caption{$\hat{u}_S^{\mathrm{rep}}(f)$}
	\label{rep}
\end{minipage}
\hfill
\begin{minipage}[b]{0.32\textwidth}
	\centering
	\includegraphics[scale=0.8, width=\textwidth]{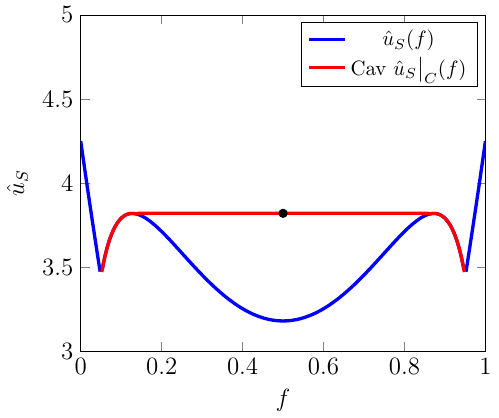}
	\caption{ $\hat{u}_S(f)$}
	\label{total}
\end{minipage}
\end{figure}

The platform's total indirect utility from announcing forecast $f$ (see Fig. \ref{total}) is therefore
\begin{equation}
\hat{u}_S(f)=\hat u_S^{\mathrm{eng}}(f)+\hat u_S^{\mathrm{rep}}(f).
\end{equation}

 The platform faces a trade-off: revealing more precise forecasts increases reputational utility but reduces user engagement. The platform's objective is to choose a calibrated forecasting strategy that maximizes its expected long-run average payoff.

Since $\mathrm{Supp}(C_\mu)=\{0.05,0.95\}$ is affinely independent, Lemma~\ref{lemma:stationarympc} and \eqref{eq:bayeaffine}
imply that the platform's maximum payoff among calibrated strategies is given by
\begin{equation}
	\mathrm{Per}(C_\mu)=\mathrm{Cav}\,\hat u_S\big|_{[0.05,0.95]}(0.5).
\end{equation}

The black dot at $f=0.5$ in Figure~\ref{total} illustrates this value, where $C=[0.05,0.95]$. For suitable values of $\kappa_1$ and $\kappa_2$, the distribution of conditionals $C_\mu$ and an optimal distribution of forecasts $Q^\star$ are given by:
\begin{equation*}
C_\mu= \begin{bmatrix}
	\lambda \\
	p 
\end{bmatrix}=
\begin{bmatrix}
	\frac{1}{2} & \frac{1}{2} \\
	\frac{5}{100}& \frac{95}{100} 
\end{bmatrix}
\quad \quad \quad \quad 
Q^\star=
\begin{bmatrix}
	\nu \\
	q
\end{bmatrix}=
\begin{bmatrix}
	\frac{1}{2} & \frac{1}{2} \\
	\frac{15}{100}& \frac{85}{100} 
\end{bmatrix}
\end{equation*}

The distribution $Q^\star$ can be induced by an optimal forecasting strategy $\sigma^\star$,
where for any $t \in \mathbb{N}$:
\begin{align*}
	\sigma^\star(f_t=15\% \mid p_t=5\%) &= \frac{8}{9}, & \sigma^\star(f_t=85\% \mid p_t=5\%) &= \frac{1}{9}, \\
	\sigma^\star(f_t=15\% \mid p_t= 95\%) &= \frac{1}{9}, & \sigma^\star(f_t=85\% \mid p_t =95\%) &= \frac{8}{9}. 
\end{align*}

The optimal calibrated strategy announces the coarse forecasts  $f_t \in \{15\%,85\%\}$ rather than the truthful conditionals  $p_t \in \{5\%,95\%\}$. This corresponds to the concave envelope of  $\hat u_S$ restricted to the  $\mathrm{Co}(\mathrm{Supp}(C_\mu))$  (red line in Fig. \ref{fig_mainresult}). More generally, as the platform places greater weight on reputation relative to engagement, its optimal forecasts become more precise. Calibration limits how far this can go: although the platform would prefer to issue extreme forecasts of $0\%$ and $100\%$, these cannot occur with positive long-run frequency, since the forecast distributions must be supported on $\mathrm{Co}(\mathrm{Supp}(C_\mu))$.

\begin{figure}[htp!]
\centering
\includegraphics[scale=0.9]{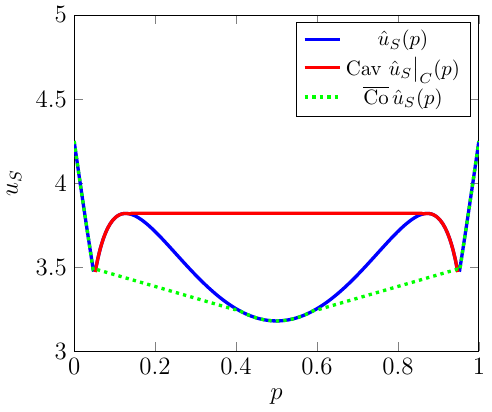}
\caption{Financial platform guarantees:
	(i) Informed sender (stationary, ergodic) - red line,
	(ii) Uninformed sender (adversarial) - green dotted line,
	(iii) Uninformed sender (ergodic Markov chain) - blue line.
}
\label{fig_mainresult}
\end{figure}

For any stochastic process, truthful forecasting is calibrated. Whenever $C_\mu$ is well defined, truthful forecasting achieves $\mathbb{E}_{C_\mu}[\hat u_S]$. In the associated persuasion problem with prior $C_\mu$, this corresponds to the full disclosure.

In the adversarial case, an uninformed platform can approximately guarantee the closed convex hull of the indirect utility, $\overline{\mathrm{Co}}\,\hat u_S$, evaluated at the realized empirical distribution of states (green dotted line). If the process is an ergodic Markov chain, this guarantee
improves to $\hat u_S(m(C_\mu))$ (blue line), again up to the grid approximation. The gap between the red and blue lines at $m(C_\mu)$ compares the informed sender's optimal payoff with the payoff guarantee for the uninformed sender, providing a benchmark for the value of information about the Markov chain.

\section{Forecasting against regret minimizers} \label{sectionnoregret}

In this section, we consider a receiver who uses  regret minimization, instead of the calibration test,  as the heuristic for  decision-making. Regret is   an exogenous criterion to evaluate a strategy in  non-Bayesian environments \citep{cesalugosi}. It measures the difference between the receiver's realized average payoff  and what he could have obtained by repeatedly playing a fixed action.   Regret-minimizing strategies ensure that, in the long run, no fixed action would have resulted in a higher average payoff. Calibration and regret minimization are closely related notions in the online learning literature \citep{vianney}. We use this connection to compare strategic forecasting under the two decision rules.

Apart from this change in the receiver's behavior, the environment and timing are as in Section \ref{sectionmodel}: in each period the sender sends a forecast, the receiver chooses an action, the state is realized, and the game moves to the next period.

We establish two main results. First, we show that a receiver who myopically best responds to calibrated forecasts incurs no regret in the long run. Second, focusing on mean-based learning algorithms, we show that the sender can attain the calibration benchmark and in some cases strictly more.

We evaluate regret conditional on each forecast, i.e., contextual regret with the forecast as the context. Intuitively, the receiver has no regret with respect to a forecast $f$ if, on the set of periods in which $f$ is sent, he cannot improve his long-run average payoff by switching to any fixed action $a^* \in A$. Let $\overline{u}_{R,T}[f]$ denote the receiver's average payoff up to period $T$ conditional on  $f$:\footnote{If $N_T[f]=0$, set $\overline{u}_{R,T}[f]=0$.}
\begin{equation*}
\overline{u}_{R,T}[f]:= \frac{1 }{N_T[f]} \sum_{t=1}^T \mathbf{1}_{\{f_t=f\}} u_R(\omega_t,a_t).
\end{equation*}

\begin{definition} 
The receiver has no regret with respect to forecast $f$ if
\begin{equation} \label{def:regret}
\limsup_{T \to \infty}  \frac{N_{T}[f]}{T} \Big( \max_{a^* \in A}u_R( \overline{\omega}_T[f],a^*)- \overline{u}_{R,T}[f] \Big) \leq 0 \quad \text{almost surely}.
\end{equation}
\end{definition}
	
The next proposition shows that if the receiver best responds to the forecasts of a calibrated strategy, that is, plays $a_t=\hat a(f_t)$ in every period, then he incurs no regret with respect to any forecast. This provides a justification for using calibration as a heuristic for decision-making.
	
\begin{proposition} \label{calibnimpliesnoreg}
Given a calibrated forecasting strategy, a receiver who myopically best responds to every forecast has no regret with respect to any forecast.
\end{proposition}
	
For a stationary ergodic process, Theorem \ref{maintheorem} characterizes the maximum payoff the informed sender can achieve under the calibration test. We
now consider a receiver who instead uses a learning algorithm.

We follow \citet{Braverman2018} and \citet{DengSivanregret} and focus on the
natural class of \textit{mean-based learning algorithms}. Mean-based algorithms
assign vanishing probability to actions whose cumulative payoff is sufficiently
below that of another action. Intuitively, mean-based strategies  play the  action that historically performs the best.   This class includes standard no-regret algorithms
such as Multiplicative Weights, Follow-the-Perturbed-Leader, and EXP3. 

We consider the $T$-period version of the game, where $T$ is large.\footnote{An infinite horizon analogue can be obtained by using, in each period $t$, the threshold $\gamma_t t$, where $\gamma_t\to0$. We focus on the finite horizon formulation used by \citet{Braverman2018} for simplicity.} We analyze the sender's expected average payoff over the $T$ periods as $T\to\infty$.

\begin{definition}
Let $U_{a,t}^f = \sum_{i=1}^t \boldsymbol{1}_{\{f_i=f\}} u_R(\omega_i,a)$ be the receiver's cumulative payoff from action $a$ over the first $t$ periods when the forecast was $f$. An algorithm is $\gamma$-mean-based if, for every forecast $f$, whenever $U_{a,t}^f<U_{b,t}^f-\gamma T$ for some $b\in A$, the probability of choosing action $a$ following forecast $f$ in period $t+1$ is at most $\gamma$. An algorithm is mean-based if it is $\gamma$-mean-based for some $\gamma=o(1)$.
\end{definition}

The following proposition shows that the sender can attain the same benchmark
as under the calibration test. To rule out degenerate ties, we assume that, within $\mathrm{Co}(\mathrm{Supp}(C_\mu))$, every receiver best response can be made strict by an arbitrarily small perturbation of the belief.\footnote{For finite state and action spaces, this condition holds generically in the receiver's payoff function.}

\begin{proposition}\label{noregretQ}
For a stationary ergodic process $\mu$ and a mean-based receiver, the sender can guarantee $\mathrm{Per}(C_\mu)$.
\end{proposition}

Recall that $\pi^\star$ is an optimal signaling policy in the static persuasion problem which uses only finitely many forecasts. As in Theorem~\ref{maintheorem}, the sender can use $\pi^\star$ in each period to generate its optimal distribution of forecasts. If the receiver has a unique best response to each forecast, that action accumulates a linear payoff advantage over every other action, since the empirical distribution of states following the forecast converges to it. A mean-based receiver therefore chooses these best responses in all but a vanishing fraction of periods, giving the sender the persuasion value. When best responses are tied, a small perturbation of the optimal signaling policy makes the sender-preferred actions strictly optimal at an arbitrarily small payoff cost.

We have shown that the sender can always do at least as well against a mean-based receiver as under the calibration benchmark. The next theorem shows that the sender can sometimes do strictly better.

\begin{theorem} \label{regretmuchmore}
 There exists a stationary ergodic process $\mu$ and a game such that, against a mean-based receiver, the sender can guarantee a payoff strictly greater than $ \mathrm{Per}(C_\mu)$.
\end{theorem}

\begin{proof} Consider the  following payoff matrix, which represents $u_S(\omega,a)$ and $u_R(\omega,a)$:
	
\begin{center}
\begin{tabular}{cc|c|c|c|c|}
	& \multicolumn{1}{c}{}  & \multicolumn{1}{c}{$a_1$} & \multicolumn{1}{c}{$a_2$}  & \multicolumn{1}{c}{$a_3$} & \multicolumn{1}{c}{$a_4$} \\\cline{3-6}
	& $0$  & $(2,8)$ &  $(0,7)$ & $(4,3)$ &  $(2,0)$  \\\cline{3-6}
	& $1$ & $(2,0)$ & $(4,3)$ & $(0,7)$ &  $(2,8)$  \\\cline{3-6}
\end{tabular}
\end{center}

The receiver's best response depends on his belief  $p = \mathbb{P}(\omega=1)$. It is optimal to play $a_1$ for $p \in [0,0.25]$,  $a_2$ for $p \in [0.25,0.5]$, $a_3$ for $p \in [0.5,0.75]$, and $a_4$ otherwise  (see Fig. \ref{utilityregret}).

\begin{figure}[htp!]
\begin{minipage}[b]{0.45\textwidth} 
	\includegraphics[width=\linewidth]{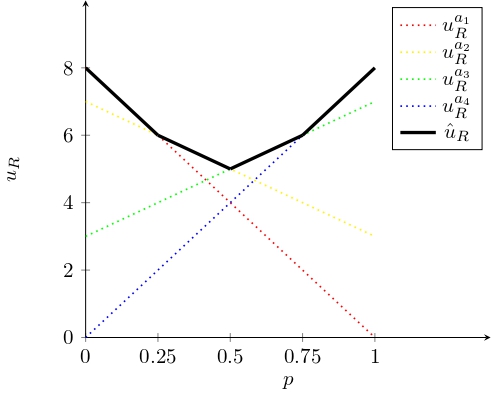}
	\caption{Receiver's indirect utility: $\hat{u}_R$} 
	\label{utilityregret}
\end{minipage}
\hfill
\begin{minipage}[b]{0.45\textwidth}
	\includegraphics[width=\linewidth]{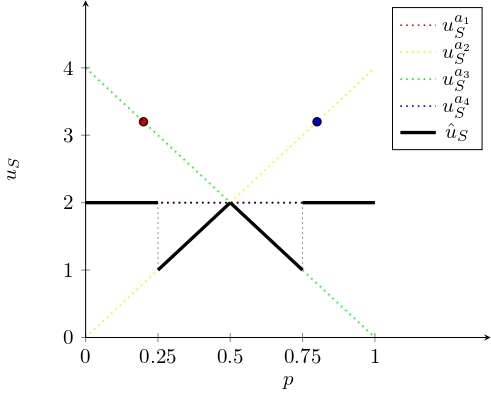}
	\caption{Sender's indirect utility: $\hat{u}_S$}
	\label{senderutilityregret}
\end{minipage}
\end{figure}

The state evolves according to a Markov chain, initialized at its invariant distribution, with transition probabilities $\mu(\omega_{t+1}=1 \mid \omega_t=1)=0.8$ and $\mu(\omega_{t+1}=1 \mid \omega_t=0)=0.2$. Hence, the  distribution  of conditionals $C_\mu$ has equal mass on $20\%$ and $80 \%$. An optimal signaling policy in the persuasion problem is full  or no disclosure, which implies $\mathrm{Per}(C_\mu)= 2$ (see Fig. \ref{senderutilityregret}).

We now describe a strategy that allows the sender to guarantee a strictly 
higher payoff against a mean-based receiver. The construction exploits the fact that mean-based algorithms adapt gradually over time.

Because the receiver uses a mean-based learning algorithm, his behavior is driven by cumulative payoffs. Conditional on forecast $f$, the cumulative payoff from action $a$ up to period $t$ is $U^f_{a,t}=N_t[f] \cdot u_R(\overline{\omega}_t[f],a)$. Thus, actions that yield higher payoff under the empirical distribution $\overline{\omega}_t[f]$ accumulate larger cumulative payoffs. Hence, once an action's cumulative payoff exceeds that of
every other action by more than $\gamma T$, the receiver
plays it with high probability.

The game has two phases, and the sender uses only two forecasts: $l$ and $h$. During  the first $\frac{T}{2}$ periods, the sender  announces $l$ when $ p_t=20\%$ and  $h$  when $p_t=80 \%$.  Under a mean-based learning algorithm,  the receiver plays $a_1$ after $l$ and $a_4$ after $h$ with high probability.  For large $T$, the sender's average payoff in this first phase is approximately $2$.\footnote{Since $\gamma=o(1)$ and the cumulative payoff difference grows at a positive rate, the additional time needed to exceed $\gamma T$ is negligible and does not affect the limiting average payoff.}

 In the remaining $\frac{T}{2}$ periods,  the sender swaps the forecasts: she announces $l$ when $p_t=80\%$ and  $h$ when $p_t=20 \%$. After this change, the empirical frequency  given $l$ 
 gradually increases from $20\%$ toward $50\%$. Once it exceeds $25\%$, 
 action $a_2$ becomes more profitable than $a_1$, and the receiver switches 
 to $a_2$ with high probability.  A similar argument applies to  $h$: 
 as the empirical frequency decreases, the receiver eventually switches 
 from $a_4$ to $a_3$.
 
For most of this second phase, the receiver   plays $a_2$ and $a_3$ 
when the  conditional distributions are $80\%$ and $20\%$, respectively. 
These actions are not best responses to the conditional beliefs and 
therefore result in payoffs outside the graph of indirect utilities 
(see the red and blue dots in Fig. \ref{senderutilityregret}). A direct calculation shows that the 
sender's average payoff in this phase is approximately $\frac{34}{11}$. Combining both phases,  the overall average payoff is approximately $\frac{28}{11} > 2$. 
 Thus, the sender achieves strictly more than the calibration benchmark.
\end{proof}

 A natural extension is to the case in which the receiver does not observe the state $(\omega_t)_{t\geq1}$ and observes only the payoff from the chosen action. Standard techniques from bandit learning, such as inverse propensity score estimators, can be used to estimate the payoffs of unplayed actions with vanishing exploration costs \citep{Bubeck,lattimore2020bandit}.

\section{Related Literature}
\label{sectionrelated}

\textbf{Strategic Information Transmission:} Our work contributes to the literature  on communication between an informed sender and an uninformed receiver. In cheap talk \citep{CrawfordSobel}, the sender's message is unverifiable, whereas in Bayesian persuasion \citep{BPKamenic}, the sender commits to how the message is generated.  In our dynamic
setup, the sender does not have commitment power and the messages are probabilistic
forecasts disciplined by calibration.   In particular, we contribute to the study of communication in dynamic environments (\citealp{RenaultSolanVieille}, \citealp{margaria2018dynamic},  \citealp{kuvalekar2022goodwill}). More specifically, the question of whether repetition can substitute for commitment is explored in \cite{pei2023repeated}, \cite{BestQuigley}, \cite{mathevet2019reputation},  and \cite{jain2026dynamic}, who offer rationales based on   lying costs, coarse summaries, reputation and  quota mechanisms, respectively. In contrast, we consider a general dynamic setting in which the receiver evaluates the sender using a calibration test.	 Our main result shows that the dynamic forecasting game can be reduced to a static persuasion problem with a simplex state space in which the sender's utility depends only on the posterior mean (\citealp{gentzkow2016rothschild}, \citealp{kolotilin2018optimal}, \citealp{DworczakPiotr}, \citealp{kleiner2021extreme}, \citealp{arieli2023optimal}).  Thus, our work shows how calibration can generate the Bayesian persuasion benchmark.	  Our paper is also related to the static forecasting model of \citet{GuoShmaya}. As in our model, messages are probabilistic forecasts. However, they define calibration error as the distance between a forecast and the true conditional distribution over states induced by the sender's strategy, and impose an exogenous cost based on this deviation. Unlike their equilibrium framework, our receiver evaluates forecasts through a calibration test based on realized forecasts and states, without forming beliefs about the sender's strategy. This restricts feasible long-run forecast distributions and leads to the persuasion benchmark. \\

\textbf{Calibration and Expert Testing:}   The initial focus of this literature has been to design a statistical test that can distinguish between an informed expert, who knows the data-generating process, and an uninformed expert, who does not. The calibration test is an objective and popular criterion for evaluating experts. However, \cite{FosterVohra} show that even an uninformed expert can produce calibrated forecasts for any process. Despite this, calibration is crucial for decision-making based on forecasts  \citep{foster2021forecast}. We build on this literature by exploring the use of the calibration test as a heuristic for decision-making.  Specifically, we characterize the extent of strategic forecasting when the calibration test is used to determine the credibility of the expert. We compare the  payoffs that an informed and  an uninformed expert can achieve for a given process. Related work studies tests of expert forecasts in decision problems, including \cite{EcheniqueShmayaFool}, \cite{gradwohlsalant2011buy} and \cite{OlszewskiPeskiPA}. These papers examine how decision-makers can benefit from expert advice while being protected against uninformed forecasts. Calibration also has important applications in machine learning (see \citealp{guptaramdas}).   We refer curious readers to \cite{foster2013calibration} and \cite{CalibrationExpertTesting}  for  comprehensive surveys on this topic. \\

\textbf{Regret minimization and Online Learning:} Regret minimization, which uses the best fixed action in hindsight as a natural benchmark, provides simple procedures for online decision-making and ensures behavior that leads to as-if rational outcomes (see \citealp{cesalugosi}, \citealp{hart2013simple}).  There has been growing interest in studying regret minimization   within strategic environments. We focus on optimizing against regret minimizing agents \citep{Braverman2018, DengSivanregret, camara2020mechanisms, lin2024generalized}. The closest paper is \cite{lin2024generalized}, which independently studies a repeated sender-receiver setting and provides non-asymptotic bounds on what the sender can guarantee. They analyze an i.i.d. setting where the sender observes the realized state in each period. In contrast, we allow dependence across periods, and the sender observes only the conditional distribution of the current state. Consequently, their bounds and techniques do not directly apply to our setup. The class of \textit{mean-based learning algorithms} was introduced by \cite{Braverman2018}
and further explored by \cite{DengSivanregret}.  A common theme in the literature   is that the sender can obtain  higher payoff  than the rational benchmark when facing  such algorithms.  Similarly, we show the sender can obtain a  higher utility than the calibration benchmark against a mean-based learner.
Additionally, \textit{U-calibration}, as introduced in \cite{kleinberg2023u}, is relevant as it ensures low external regret for receivers with unknown payoff functions.

\section{Conclusion} \label{sectionconclusion}
We study strategic forecasting when credibility is determined by calibration. Calibration constrains, but does not eliminate, the sender's ability to forecast strategically. For stationary and ergodic processes, calibration allows the sender to strategically coarsen truthful forecasts: the distributions that can arise under calibration are  mean-preserving contractions of the distribution of truthful forecasts. This reduces the sender's dynamic optimization problem to a static persuasion problem. Comparing informed and uninformed senders further separates credibility from information: an expert need not know the data-generating process to produce calibrated forecasts. However, information allows the sender to benefit more from strategic forecasting. Finally, our results on regret minimization show that replacing calibration with a different decision-making heuristic does not necessarily eliminate the sender's ability to benefit from strategic forecasts.

Many problems remain open for the setting that we study. In particular, characterizing the optimal calibrated strategy for an arbitrary stochastic process remains open. Since this problem may be intractable in full generality, a natural question is for which classes of stochastic processes one can find a calibrated strategy that does better than truthful forecasting. A complete characterization of attainable payoffs against no-regret learners also remains open.

The model admits natural extensions. The receiver may use alternative statistical tests to assess credibility. This raises the question of which tests make truthful forecasting optimal and which limit the scope for strategic misreporting. One can also let the receiver's actions affect state transitions\textemdash for example, in a Markov decision process. An interesting question then is to characterize which distributions of forecasts are feasible under calibration in that setting.

\bibliographystyle{plainnat}
\bibliography{EC}

\appendix

\section{Omitted Results and  Proofs} \label{sec:proofs}

\subsection{Proposition \ref{prop:slowerror}}

\begin{proposition} \label{prop:slowerror}
Fix a stationary ergodic process $\mu$. There exists a sequence
$(\epsilon_t)_{t\geq1}$ with $\epsilon_t\to0$ such that every calibrated
strategy induced by a fixed rule $\pi:D\to\Delta(F_0)$, with $F_0$ finite,
fails the calibration test only finitely many times almost surely.
\end{proposition}

\begin{proof}
Fix a calibrated forecasting strategy $\sigma$ of the stated form, and let
$F_0$ denote its finite forecast support. For each $f\in F_0$,
\begin{align}
\frac{N_T[f]}{T}\bigl(\overline \omega_T[f]-f\bigr)
={}&
\frac1T\sum_{t=1}^T
\mathbf{1}_{\{f_t=f\}}(\delta_{\omega_t}-p_t) \label{eq:decomp1}\\
&+
\frac1T\sum_{t=1}^T
\bigl(\mathbf{1}_{\{f_t=f\}}-\pi(f\mid p_t)\bigr)(p_t-f)
\label{eq:decomp2}\\
&+
\frac1T\sum_{t=1}^T
\pi(f\mid p_t)(p_t-f). \label{eq:decomp3}
\end{align}
The first two terms are averages of bounded martingale differences and
converge to zero almost surely. Since $\sigma$ is calibrated, the norm of
the left-hand side also converges to zero for every $f\in F_0$. Hence, the
third term converges to zero almost surely. Since $\mu$ is stationary and
ergodic, Lemma~\ref{stationaryergodic} implies that
$N_T[p]/T\to C_\mu(p)$ almost surely for every $p\in D$. It follows that
\begin{equation} \label{eq:balancepi}
\sum_{p\in D} C_\mu(p)\pi(f\mid p)(p-f)=0
\qquad\text{for every } f\in F_0.
\end{equation}

By Lemma~3 of \cite{foster2011complexity}, applied to the two bounded martingale-difference sequences above,  there exist constants $z_1,z_2>0$ such that,  for every $T\geq 1$, every $f\in F_0$, and every $\theta>0$,
\[
\mathbb P\left(
\left\|
\frac1T\sum_{t=1}^T
\mathbf 1_{\{f_t=f\}}(\delta_{\omega_t}-p_t)
\right\|
\geq\theta
\right)
\leq 2e^{-z_1T\theta^2},
\]
and
\[
\mathbb P\left(
\left\|
\frac1T\sum_{t=1}^T
\bigl(\mathbf 1_{\{f_t=f\}}-\pi(f\mid p_t)\bigr)(p_t-f)
\right\|
\geq\theta
\right)
\leq 2e^{-z_2T\theta^2}.
\]
Let $z:=\min\{z_1,z_2\}$, and let $M_T$ denote the sum over $f\in F_0$
of the norms of the first two terms in
\eqref{eq:decomp1}--\eqref{eq:decomp2}. If $M_T>T^{-1/4}$, then at least
one of these $2|F_0|$ nonnegative terms must exceed
$T^{-1/4}/(2|F_0|)$. Hence, by a union bound,
\[
\mathbb P\left(M_T>T^{-1/4}\right)
\leq
4|F_0|
\exp\left(
-\frac{z\sqrt{T}}{4|F_0|^2}
\right).
\]
The right-hand side is summable in $T$. Therefore, by the Borel--Cantelli
lemma,
\begin{equation} \label{eq:MTbound}
M_T\leq T^{-1/4}
\end{equation}
for all sufficiently large $T$, almost surely.

Next, define
\[
R_T:=
\sum_{p\in D}
\left|
\frac{N_T[p]}{T}-C_\mu(p)
\right|.
\]
Since $\mu$ is stationary and ergodic,
Lemma~\ref{stationaryergodic} implies that $R_T\to0$ almost surely.
Hence, there exists a deterministic sequence $b_T\to0$, depending only
on $\mu$, such that
\begin{equation} \label{eq:RTbound}
R_T\leq b_T
\end{equation}
for all sufficiently large $T$, almost surely.\footnote{To see this, for
each $k$ choose $T_k$ such that
$\mathbb{P}(\sup_{T\geq T_k}R_T>1/k)\leq 2^{-k}$, with $T_k\to\infty$,
and set $b_T=1/k$ for $T_k\leq T<T_{k+1}$. The claim then follows from
the Borel--Cantelli lemma.}

Using \eqref{eq:balancepi}, for every $f\in F_0$,
\begin{align*}
\frac1T\sum_{t=1}^T\pi(f\mid p_t)(p_t-f)
&=
\sum_{p\in D}\frac{N_T[p]}{T}\pi(f\mid p)(p-f)\\
&=
\sum_{p\in D}
\left(\frac{N_T[p]}{T}-C_\mu(p)\right)
\pi(f\mid p)(p-f).
\end{align*}
Since $\|p-f\|\leq\sqrt{2}$ and
$\sum_{f\in F_0}\pi(f\mid p)=1$, summing over $f\in F_0$ gives
\[
\sum_{f\in F_0}
\left\|
\frac1T\sum_{t=1}^T\pi(f\mid p_t)(p_t-f)
\right\|
\leq \sqrt{2}\,R_T.
\]
Therefore, by \eqref{eq:decomp1}--\eqref{eq:decomp3}, the calibration
error satisfies
\[
\sum_{f\in F_0}
\frac{N_T[f]}{T}\|\overline \omega_T[f]-f\|
\leq M_T+\sqrt{2}\,R_T.
\]

Finally, set
\[
\epsilon_T:=T^{-1/4}+\sqrt{2}\,b_T.
\]
Then $\epsilon_T\to0$, and
\eqref{eq:MTbound}--\eqref{eq:RTbound} imply that the calibration error is
at most $\epsilon_T$ for all sufficiently large $T$, almost surely.
Hence, $\sigma$ fails the calibration test only finitely many times almost
surely.
\end{proof}

\subsection{Proof of Lemma \ref{lemma:stationarympc}}

(i) First, we show that for a stationary ergodic process $\mu$ and any calibrated
strategy $\sigma$, $\mathcal F_{\mu,\sigma}\subseteq \mathcal M(C_\mu)$.  The idea is to group forecasts into finitely many bins. For each partition, calibration implies that the resulting finite distribution is a simple mean-preserving contraction of $C_\mu$. Letting the bin size go to zero completes the proof.

For each period $T$, define the empirical frequencies
\[
N_T[p]=\sum_{t=1}^T \boldsymbol{1}_{\{p_t=p\}},
\qquad
N_T[f]=\sum_{t=1}^T \boldsymbol{1}_{\{f_t=f\}},
\qquad
N_T[f,p]=\sum_{t=1}^T \boldsymbol{1}_{\{f_t=f,\,p_t=p\}}.
\]

Fix in advance a sequence of finite partitions of the set of forecasts  whose maximum cell diameters converge to zero. Since there are only countably many cells, the martingale convergence used below holds simultaneously for all of them almost surely.

Let $Q\in\mathcal F_{\mu,\sigma}$ be any limit point of the empirical
distribution of forecasts, and let $(T_k)$ be a subsequence along which
the empirical distribution converges to $Q$. Fix a partition
$B_1,\ldots,B_m$ from this sequence, with maximum diameter $\eta$.

For each cell, let
\[
N_T[B_j]=\sum_{f\in B_j}N_T[f],
\qquad
N_T[B_j,p]=\sum_{f\in B_j}N_T[f,p],
\]
and define
\[
G_T[p,B_j]=\frac{N_T[B_j,p]}{N_T[p]}
\quad\text{if }N_T[p]>0,
\]
and arbitrarily otherwise. Then, for every $B_j$,
\begin{equation}\label{eq:mass-bin}
\sum_{p\in D}\frac{N_T[p]}{T}G_T[p,B_j]
=
\frac{N_T[B_j]}{T}.
\end{equation}

It remains to verify mean preservation. Whenever $N_T[B_j]>0$, let
\[
\overline f_T[B_j]
=
\frac{1}{N_T[B_j]}\sum_{t:f_t\in B_j}f_t,
\qquad
\overline p_T[B_j]
=
\frac{1}{N_T[B_j]}\sum_{t:f_t\in B_j}p_t,
\]
and define $\overline\omega_T[B_j]$ analogously. By the triangle inequality,
\begin{equation}\label{eq:mean-bin}
\frac{N_T[B_j]}{T}
\|\overline f_T[B_j]-\overline p_T[B_j]\|
\leq
\frac{N_T[B_j]}{T}
\|\overline f_T[B_j]-\overline\omega_T[B_j]\|
+
\frac{N_T[B_j]}{T}
\|\overline\omega_T[B_j]-\overline p_T[B_j]\|.
\end{equation}
The first term converges to zero almost surely by calibration, since
\[
\frac{N_T[B_j]}{T}
\|\overline f_T[B_j]-\overline\omega_T[B_j]\|
\leq
\sum_{f\in B_j}\frac{N_T[f]}{T}
\|\overline\omega_T[f]-f\|
\to0.
\]
The second term can be rewritten as an average of bounded martingale
differences,
\[
\left\|\frac1T\sum_{t=1}^T X_t^j\right\|,
\qquad
X_t^j=
\boldsymbol{1}_{\{f_t\in B_j\}}
(\delta_{\omega_t}-p_t).
\]
Hence, by the strong law of large numbers for bounded martingale differences,
it also converges to zero. Thus, the left-hand side of
\eqref{eq:mean-bin} converges to zero almost surely.

Finally, since $\mu$ is stationary and ergodic,
$N_T[p]/T\to C_\mu(p)$ almost surely. Since there are finitely many
$p\in D$ and cells $B_j$, along a further subsequence we may assume that
\[
G_{T_k}[p,B_j]\to G[p,B_j]
\]
for every $p$ and $j$. Since each $G_{T_k}$ is row stochastic, so is $G$.

By \eqref{eq:mass-bin}, along $(T_k)$, the empirical frequency
of each cell converges to
\[
q_j:=\sum_{p\in D}C_\mu(p)G[p,B_j].
\]
For every cell with $q_j>0$, the average conditional
$\overline p_{T_k}[B_j]$ converges, since the empirical frequencies
of each pair $(p,B_j)$ converge. By \eqref{eq:mean-bin}, the average
forecast $\overline f_{T_k}[B_j]$ converges to the same limit;
denote it by $f_j^\eta$.

Define
\[
Q^\eta:=\sum_{j:q_j>0}q_j\delta_{f_j^\eta}.
\]

Taking to the limit in \eqref{eq:mass-bin} and \eqref{eq:mean-bin},  $G$
satisfies the mass- and mean-preservation conditions \eqref{smpc}--\eqref{smpc2} for $C_\mu$ and $Q^\eta$. Hence,
$Q^\eta$ is a simple mean-preserving contraction of $C_\mu$.

Along the chosen subsequence, $Q^\eta$ is the limit of the empirical forecast distributions obtained by replacing forecasts within each cell by their average. Since each
cell has diameter at most $\eta$, this changes each forecast by at most $\eta$. Therefore, $Q^\eta\to Q$ weakly as $\eta\to0$. Hence, $Q\in\mathcal M(C_\mu)$, and therefore $\mathcal F_{\mu,\sigma}\subseteq\mathcal M(C_\mu)$.

(ii) Next, we show that if $\mu$ is stationary and ergodic, then for any
$Q\in\mathcal{M}(C_\mu)$ with finite support, there exists a calibrated strategy $\sigma$ that
induces $Q$, that is, $F_{\mu,\sigma}=Q$ and  fails the calibration test with expected long-run frequency zero.

  Recall that the informed sender knows the  distribution of conditionals $C_\mu$. Fix $Q \in \mathcal{M} (C_\mu)$ and let $\pi$ denote  the  signaling policy that implements $Q$. The informed sender can compute the  conditional  $p_t = \mu(\cdot \mid \omega^{t})$ at the start of each period $t$. Define the  forecasting strategy $\tilde \sigma$  by  $\tilde \sigma_t(f_t =f \mid p_t=p)= \pi(f\mid p)$ for all $t$. We  show the  strategy $ \tilde \sigma$ is calibrated. For any $f \in \mathrm{Supp}(Q)$,
 \begin{align*}
 	\frac{N_T[f]}{T}\|\overline{\omega}_T[f]-f\|
 	&\le
 	\left\|\frac1T\sum_{t=1}^T \mathbf 1_{\{f_t=f\}}(\delta_{\omega_t}-p_t)\right\|
 	+\left\|\frac1T\sum_{t=1}^T(\mathbf 1_{\{f_t=f\}}-\pi(f\mid p_t))(p_t-f)\right\| \\
 	&\quad+
 	\left\|\frac1T\sum_{t=1}^T \pi(f\mid p_t)(p_t-f)\right\|.
 \end{align*}
 
The first two terms are averages of bounded martingale differences and therefore
converge to zero almost surely. By stationarity and ergodicity, the third term converges almost surely to $\sum_{p} C_\mu(p)\pi(f\mid p)(p-f)$, which equals zero since $\pi$ implements $Q\in\mathcal M(C_\mu)$. 

Summing over the finite $f\in \mathrm{Supp}(Q)$, we conclude that $\tilde \sigma$ is calibrated. Moreover,  by Assumption \ref{assumerror}, it fails the  test only in expected long-run frequency zero of periods.

Finally, we show that  $\tilde \sigma$ induces the distribution of forecasts $Q$. Given
$f\in \mathrm{Supp}(Q)$, 
\[
\frac1T\sum_{t=1}^T \mathbf 1_{\{f_t=f\}}
=\frac1T\sum_{t=1}^T \pi(f\mid p_t)
+\frac1T\sum_{t=1}^T\bigl(\mathbf 1_{\{f_t=f\}}-\pi(f\mid p_t)\bigr).
\]
The second term goes to zero since it is an average of bounded martingale differences. Moreover,
\[
\frac1T\sum_{t=1}^T \pi( f \mid p_t)
\to \sum_{p\in D} C_\mu(p)\pi(f\mid p)=Q(f),
\]
where the convergence follows from Lemma \ref{stationaryergodic} and the equality holds because $\pi$ implements $Q$. Hence,
$N_T[f]/T \to Q(f)$ for all $f \in \mathrm{Supp}(Q)$, and therefore
$F_{\mu,\tilde\sigma}=Q$.

\subsection{Proof of Theorem \ref{maintheorem}}

We first establish the upper bound. Let $I_t$ indicate whether the sender passes the calibration test in period $t$. Let
\[
d_T:=\frac1T\sum_{t=1}^T(1-I_t)
\]
denote the fraction of failure periods up to $T$.

For each forecast $f$, let $N_T^+[f]$ and $N_T^-[f]$ denote the number of periods
up to $T$ in which $f$ is sent and the test is passed and failed, respectively.
Let $\overline{\omega}_T^+[f]$ and $\overline{\omega}_T^-[f]$ denote the corresponding
empirical distributions of states. Thus,
\[
N_T[f]=N_T^+[f]+N_T^-[f],
\qquad
d_T=\sum_f\frac{N_T^-[f]}{T}.
\]

Consider a horizon $T$ at which the sequence of forecasts up to $T$ is
$\epsilon_T$-calibrated. Since
\[
\sum_f\frac{N_T[f]}{T}
\|\overline{\omega}_T[f]-f\|
\leq\epsilon_T,
\]
and
\[
N_T[f](\overline{\omega}_T[f]-f)
=
N_T^+[f](\overline{\omega}_T^+[f]-f)
+
N_T^-[f](\overline{\omega}_T^-[f]-f),
\]
the triangle inequality and
$\|\overline{\omega}_T^-[f]-f\|\leq\sqrt{2}$ imply
\[
\sum_f\frac{N_T^+[f]}{T}
\|\overline{\omega}_T^+[f]-f\|
\leq
\epsilon_T+\sqrt{2}d_T.
\]

The sender's average payoff up to $T$ is
\[
\frac1T\sum_{t=1}^T u_S(\omega_t,a_t)
=
\sum_f\frac{N_T^+[f]}{T}
u_S(\overline{\omega}_T^+[f],\hat a(f))
-cd_T.
\]
Since
\[
u_S(\overline{\omega}_T^+[f],\hat a(f))
\leq
\hat u_S(f)
+
L\|\overline{\omega}_T^+[f]-f\|,
\]
we obtain
\[
\frac1T\sum_{t=1}^T u_S(\omega_t,a_t)
\leq
\sum_f\frac{N_T^+[f]}{T}\hat u_S(f)
+
L\epsilon_T
+
(\sqrt{2}L-c)d_T.
\]

Since $N_T^+[f]=N_T[f]-N_T^-[f]$ and $|\hat u_S(f)|\leq L$,
\[
\sum_f\frac{N_T^+[f]}{T}\hat u_S(f)
\leq
\sum_f\frac{N_T[f]}{T}\hat u_S(f)+Ld_T.
\]
Therefore,
\[
\frac1T\sum_{t=1}^T u_S(\omega_t,a_t)
\leq
\sum_f\frac{N_T[f]}{T}\hat u_S(f)
+
L\epsilon_T
-
\kappa d_T,
\]
where
\[
\kappa:=c-(1+\sqrt{2})L>0
\]
by Assumption~\ref{assum:highcost}.

Along any sequence of such horizons, $\epsilon_T\to0$. By the same argument as
in the proof of Lemma~\ref{lemma:stationarympc}, every limit point of the empirical
distribution of forecasts belongs to $\mathcal M(C_\mu)$. Since $\hat u_S$ is
upper semicontinuous under the maintained tie-breaking convention, it follows that
\begin{equation}
\limsup_{T\to\infty}
\left[
\frac1T\sum_{t=1}^T u_S(\omega_t,a_t)
+\kappa d_T
\right]
\leq
\operatorname{Per}(C_\mu)
\label{eq:passing-upper}
\end{equation}
along horizons at which the sequence of forecasts is $\epsilon_T$-calibrated. Since $\kappa>0$, this also implies that the sender's average payoff is bounded above by the persuasion value. We keep the
stronger inequality because it will also imply that every optimal strategy fails the
calibration test in a vanishing fraction of periods in expectation.

We now extend this bound to all horizons. Define
\[
Z_T
:=
\sum_{t=1}^T u_S(\omega_t,a_t)
+
\kappa\sum_{t=1}^T(1-I_t).
\]
If the sender fails the test in period $t$, then
\[
Z_t-Z_{t-1}
=
-c+\kappa
=
-(1+\sqrt{2})L.
\]
Since $\operatorname{Per}(C_\mu)\geq-L$,
\[
Z_t-Z_{t-1}
\leq
\operatorname{Per}(C_\mu).
\]
Thus, during failure periods,
\[
Z_t-t\operatorname{Per}(C_\mu)
\]
cannot increase.

If the sender passes the test in period $t$, then the sequence of forecasts up to
$t-1$ is $\epsilon_{t-1}$-calibrated, so \eqref{eq:passing-upper} applies at
$t-1$. Since one additional period changes $Z_t$ by only a bounded amount, it follows that
\[
\limsup_{T\to\infty}\frac{Z_T}{T}
\leq
\operatorname{Per}(C_\mu)
\qquad\text{almost surely}.
\]
If there are only finitely many passing periods, the same conclusion follows
immediately since thereafter every period is a failure.

Since $Z_T/T$ is uniformly bounded, the reverse Fatou lemma gives
\begin{equation}
\limsup_{T\to\infty}
\mathbb E_{\sigma,\mu}
\left[
\frac1T\sum_{t=1}^T u_S(\omega_t,a_t)
+\kappa d_T
\right]
\leq
\operatorname{Per}(C_\mu).
\label{eq:strong-upper}
\end{equation}
Since $\kappa d_T\geq0$, the sender's expected long-run payoff under any forecasting
strategy is at most $\operatorname{Per}(C_\mu)$.

We now show that this bound is attained. Let
\[
Q^\star\in
\arg\max_{Q\in\mathcal M(C_\mu)}
\mathbb E_Q[\hat u_S].
\]
Since $\operatorname{Supp}(C_\mu)$ is finite, $Q^\star$ can be chosen with finite
support. Let $\pi^\star$ be a signaling policy that implements $Q^\star$, and define
the forecasting strategy
\[
\sigma_t^\star(\cdot\mid p_t)
=
\pi^\star(\cdot\mid p_t)
\]
in every period.

By Lemma~\ref{lemma:stationarympc}, this strategy is calibrated, induces
$Q^\star$, and fails the calibration test with expected long-run frequency zero.
Consider the benchmark in which the receiver plays $\hat a(f_t)$ in every period.
By the same martingale and ergodic arguments as in Lemma~\ref{lemma:stationarympc},
the average benchmark payoff converges almost surely to
\[
\sum_p C_\mu(p)
\sum_f\pi^\star(f\mid p)
\sum_\omega p(\omega)u_S(\omega,\hat a(f))
=
\sum_fQ^\star(f)\hat u_S(f)
=
\operatorname{Per}(C_\mu).
\]
Since payoffs are bounded, the benchmark payoff also converges in expectation to
$\operatorname{Per}(C_\mu)$. By Lemma~\ref{lemma:stationarympc}, failures
occur in a vanishing expected fraction of periods, so they do not affect the expected
long-run average payoff. Hence, the strategy attains $\operatorname{Per}(C_\mu)$, which completes the proof.

Finally, \eqref{eq:strong-upper} implies that every optimal forecasting strategy
fails the calibration test in a vanishing fraction of periods in expectation.
Indeed, for an optimal strategy $\sigma^\star$, the expected average payoff converges to
$\operatorname{Per}(C_\mu)$. Since $\kappa>0$, it follows that
$\mathbb E_{\sigma^\star,\mu}[d_T]\to0$.

\subsection{Proposition \ref{proposition:mpc}}

	\begin{proposition} \label{proposition:mpc}
Fix an arbitrary stochastic process $\mu$. Let $\sigma$ be a forecasting strategy whose forecasts take values in a finite set $F_0 \subset F$. Then $\sigma$ is calibrated if and only if
		\begin{equation} \label{mpcproofgeneral}
			\limsup_{T \to \infty} \sum_{f \in F_0} \frac{N_T[f]}{T} \left\| f - \sum_{p \in D} p \, \frac{N_T[f,p]}{N_T[f]} \right\| = 0 \quad \mathbb{P}_{\sigma, \mu}\text{-a.s.}
		\end{equation}
		where $	N_T[f,p] = \sum_{t=1}^T \mathbf{1}_{\{ f_t = f,\, p_t=p\}}$.\footnote{If $N_T[f]=0$ in \eqref{mpcproofgeneral}, we set the entire summand equal to zero.}
	\end{proposition}

\begin{proof}
($\Rightarrow$) First, we show that if the forecasting strategy $\sigma$ is calibrated, then  $\eqref{mpcproofgeneral}$ holds. Using the triangle inequality, for any $f \in F_0$, 
\begin{equation*}
 \frac{N_T[f]}{T} \|f- \sum_{p \in D} p \frac{N_T[f,p]}{N_T[f]} \|   
\leq   \frac{ N_T[f]}{T} \|f-   \overline{\omega}_T[f] \| +  \| \frac{1}  {T} \sum_{t=1}^T  \boldsymbol{1}_{\{f_t=f \}} \bigl(\delta_{\omega_t}-p_t\bigr)  \|. 
\end{equation*}

The first term converges to zero since $\sigma$ is calibrated. For the second term, define
\[
X_t:=\mathbf 1_{\{f_t=f\}}\bigl(\delta_{\omega_t}-p_t\bigr).
\]
Then $(X_t)_{t\ge 1}$ is a bounded martingale difference sequence with respect to the natural filtration. Hence, by the strong law of large numbers for bounded martingale differences,
\[
\frac1T\sum_{t=1}^T X_t \to 0
\qquad \mathbb P_{\sigma,\mu}\text{-a.s.}
\]
Since $F_0$ is finite, we may take the intersection of these
probability one events over all $f \in F_0$. Therefore,
\[
\limsup_{T\to\infty} \sum_{f \in F_0}\frac{N_T[f]}{T} \|f- \sum_{p \in D} p \frac{N_T[f,p]}{N_T[f]} \|=0 \quad \mathbb{P}_{\sigma,\mu}\text{-a.s.}
\]

($\Leftarrow$) Conversely, suppose that \eqref{mpcproofgeneral} holds.
Fix $f \in F_0$. Using the triangle inequality,
\begin{equation*} 
 \frac{ N_T[f]}{T} \|f-    \overline{\omega}_T[f] \| \leq   \frac{N_T[f]}{T} \|f- \sum_{p} p \frac{N_T[f,p]}{N_T[f]} \|    +  \| \frac{1}{T}  \sum_{t=1}^T  \boldsymbol{1}_{\{f_t=f \}}  \bigl(\delta_{\omega_t}-p_t\bigr)   \|. 
\end{equation*}

The first term converges to zero by assumption, and the second
converges to zero by the same martingale argument as above.
Since $F_0$ is finite, summing over $f \in F_0$ yields
\[
\limsup_{T\to\infty}
\sum_{f\in F_0}
\frac{N_T[f]}{T}
\| f - \overline{\omega}_T[f] \|
=0
\quad
\mathbb P_{\sigma,\mu}\text{-a.s.}
\]
\end{proof}

\subsection{Lemma \ref{stationaryergodic}}

	A  stochastic process $(\omega_t)_{t \geq 1}$ is \textit{stationary} if, for any $k \in \mathbb{N}$, the joint distribution of the $k$-tuple ($\omega_t,\omega_{t+1},\ldots,\omega_{t+k-1}$) does not depend on $t$.   Let $T:  \Omega^\infty \to \Omega^\infty  $ be the shift transformation given by $T(\omega)_t = \omega_{t+1}$ for all $t \in \mathbb{N}$. Let $\mathcal{I}$ denote the $\sigma-$algebra of all invariant Borel sets for the transformation $T$. The stationary process  $(\omega_t)_{t \geq 1}$ is \textit{ergodic} if $\mathcal{I}$ is trivial, that is, $\mathbb{P}(A) \in \{0,1 \}$ for all $A \in \mathcal{I}$.

\begin{lemma} \label{stationaryergodic}
For a stationary ergodic process $\mu$, the  distribution of conditionals  $C_\mu$ exists and is constant  almost surely.
\end{lemma}

\begin{proof} Consider the two-sided extension of the stationary process so that  $\mu \in \Delta( \Omega^{\mathbb{Z}}).$  Let $\omega_a^b=(\omega_a,\ldots,\omega_{b-1})$ where $a < b$.  Define
\begin{equation*}
f_n := \mu(\omega_0=\cdot \mid \omega_{-n}^0), \qquad f_\infty := \mu(\omega_0=\cdot \mid \omega_{-\infty}^0).
\end{equation*}

By the martingale convergence theorem, $f_n \to f_\infty$ almost surely. By stationarity,
\[
p_n:=\mu(\omega_n=\cdot \mid \omega_0^n)= f_n\circ T^n.
\]
where $T$ is the  shift transformation.  Recall that $p_n$ takes values in the finite set $D\subset\Delta(\Omega)$; since $\mu$ is $T$-invariant and $p_n=f_n\circ T^n$, it follows that $f_n$ also takes values in $D$.

 Fix $p\in D$ and let $g_n=\mathbf{1}_{\{f_n=p\}}$ and $g_\infty=\mathbf{1}_{\{f_\infty=p\}}$. Since $f_n\to f_\infty$ almost surely and $D$ is finite, we
have $g_n\to g_\infty$ almost surely, and $\sup_n |g_n|\le 1\in L_1(\mu)$. Maker's
ergodic theorem  (see \citealp{Kallenberg2002FoMP})  implies
\[
\lim_{N\to\infty}\frac{1}{N}\sum_{n=0}^{N-1} g_n\circ T^n
= \mathbb E[g_\infty \mid \mathcal{I}]
\qquad \mathbb P_\mu\text{-a.s.},
\]
where $\mathcal I$ is the shift-invariant $\sigma$-algebra. Since
$g_n\circ T^n=\mathbf{1}_{\{p_n=p\}}$, this becomes
\[
\lim_{N\to\infty}\frac{1}{N}\sum_{n=0}^{N-1}\mathbf{1}_{\{p_n=p\}}
= \mathbb E[\mathbf{1}_{\{f_\infty=p\}}\mid\mathcal I]
\qquad  \mathbb P_\mu\text{-a.s.}
\]
If $\mu$ is ergodic then $\mathcal I$ is trivial, so the right-hand side equals
 $\mathbb E[\mathbf{1}_{\{f_\infty=p\}}]$, which is constant. Hence, for each $p\in D$, the limit
\[
C_\mu(p):=\lim_{N\to\infty}\frac{1}{N}\sum_{n=0}^{N-1}\mathbf{1}_{\{p_n=p\}}
\]
exists and is  constant almost surely.
\end{proof}

\subsection{Proof of Theorem \ref{maxattainable} }

The main idea of the proof is to combine the sender's payoff and  calibration cost  to form a vector-valued payoff. We then identify a suitable closed and convex target set that captures the definition of attainable function (a similar approach appears in \citealp{mannorTsitsiklisYu}).  Finally, we apply the dual condition of approachability to demonstrate that the set is approachable, ensuring convergence regardless of nature's strategy.

Given a forecast $f \in F_\epsilon$ and state $\omega \in \Omega$, the calibration cost vector is given by
\begin{equation}
c(f,\omega)=(\boldsymbol{0},\ldots,f- \delta_{\omega},....,\boldsymbol{0}) \in \mathbb{R}^{|F_\epsilon||\Omega|}.
\end{equation}

It is  a vector of $|F_\epsilon|$ elements of size $\mathbb{R}^{|\Omega|}$ with one non-zero element (at the position for $f$) while the rest are equal to $\boldsymbol{0} \in \mathbb{R}^{|\Omega|}$.

In each period $t$, the sender and nature simultaneously choose $f_t \in F_\epsilon$ and $\omega_t \in \Omega$ respectively, resulting  in  payoff $r_t= u_S( \omega_t, \hat a (f_t))$ and calibration cost $c_t=c(f_t,\omega_t)$. The $\epsilon$-calibration condition \eqref{eq:calibrationasym} can be rewritten as follows:  the  average of the sequence of  vector-valued calibration costs $c_t=c(f_t,\omega_t)$ converges almost surely to the set $\mathcal{C}$, where\footnotemark  
$$\mathcal{C}= \{ c \in \mathbb{R}^{|F_\epsilon| |\Omega|} : \sum_{f \in F_\epsilon} \|\underline{c}_f \| \leq \epsilon \}. $$ 

\footnotetext{ Formally, this implies    $\lim_{T \to \infty} \mathrm{dist}(\frac{1}{T}\sum_{t=1}^T c(f_t,\omega_t), \mathcal{C})= 0$ almost surely,  where $\mathrm{dist}(.)$ is the Euclidean distance.}

We now show the function $ p \mapsto \overline{\mathrm{Co}} \,\hat u_S^\epsilon$ is attainable. For each  period $t$, consider the vector-valued  payoff
\begin{equation}
m(f_t,\omega_t):= (r_t,c_t, \delta_{\omega_t}) \in  \mathbb{R} \times \mathbb{R}^{|F_\epsilon||\Omega|} \times \Delta (\Omega).
\end{equation}

Define the sets:
\begin{equation*}
\mathcal T_1 =\{(r,c,p) : r \geq \overline{\mathrm{Co}}  \, \hat u_S^\epsilon(p)\}, \quad	\mathcal  T_2 =\{(r,c,p) : c \in \mathcal{C} \},
\end{equation*}
and let  $\mathcal T = \mathcal T_1 \cap \mathcal T_2$ denote the target set. The set $\mathcal T$ is closed and convex.

To show that the average of $m(f_t,\omega_t)$ approaches $\mathcal T$, it suffices to
verify the dual condition of approachability \citep{blackwell1956analog},
namely that for every $p\in\Delta(\Omega)$ there exists
$x\in\Delta(F_\epsilon)$ such that $m(x,p)\in \mathcal{T}$, where
\begin{equation}\label{dual}
	m(x,p)
	=\Big(
	\sum_{f\in F_\epsilon} x(f) \sum_{\omega \in \Omega} p(\omega) u_S(\omega,\hat{a}(f)) ,\,
	\sum_{f\in F_\epsilon}\sum_{\omega\in\Omega} x(f)p(\omega)c(f,\omega),\,
	p
	\Big).
\end{equation}

Fix \( p \in \Delta(\Omega) \) and choose
\(
x = \delta_{f^*(p)} \in \Delta(F_\epsilon),
\)
that is, the degenerate distribution that assigns probability one to the
forecast \( f^*(p) \in F_\epsilon \). We verify that this choice satisfies
the dual condition.

With this choice, the  sender's payoff equals $\hat u^\epsilon_S(p)$. Hence, by the definition of closed convex hull, it follows that 
$m(x,p)\in \mathcal T_1$. Moreover, $\sum_{\omega} p(\omega)c(f^*(p),\omega)
=(\boldsymbol{0},\dots,f^*(p)-p,\dots,\boldsymbol{0}),$
so the calibration error equals $\|f^*(p)-p\|\le \epsilon$, and hence
$m(x,p)\in \mathcal T_2$. Therefore $m(x,p)\in \mathcal T$, the dual condition holds, and the set $\mathcal T$ is approachable.

Finally, we show that approachability of $\mathcal T$ implies the two conditions of attainability. Since $F_\epsilon$ is finite, $\hat u_S^\epsilon$ is piecewise affine, and its closed convex
envelope is continuous. Thus, approaching $\mathcal T$ ensures that the sender's  average payoff  exceeds $\overline{\mathrm{Co}}\,\hat u_S^\epsilon(\overline{\omega}_T)$ while remaining $\epsilon$-calibrated.

\subsection{Proof of Proposition \ref{uninformedstationary}}

We use the notion of opportunistic approachability, introduced in
\citet{bernsteinmanorshimkin}. The strategies proposed there not only ensure approachability of the target set in the adversarial case, but also exploit restrictions in nature's play to approach smaller subsets when such restrictions arise, either in a statistical or an empirical sense; we focus on the latter. 

\begin{definition}[\cite{bernsteinmanorshimkin}]
	The play of nature is \textit{empirically $Q$-restricted} with respect to a partition
	$\{l_m\}$, if there exists a convex subset $Q \subseteq  \Delta (\Omega)$ such that, 
	\[
	\lim_{M \rightarrow \infty} \frac{1}{n_M} \sum_{m=1}^M l_m \,\text{dist}(\overline{\omega}_m, Q )= 0, \quad \mathbb{P}_{\mu}\text{-a.s. }
	\]
	where $l_m$ denotes the length of block $m$, $n_M =\sum_{m=1}^M l_m$ denotes the time at
	the end of block $M$, and $\overline{\omega}_m$ denotes the empirical distribution of states
	in block $m$.
\end{definition}

Apply the opportunistic approachability strategy to the vector payoff $m$ used in the proof of Theorem~\ref{maxattainable}. If nature's play is empirically $Q$-restricted with respect to a suitable polynomially growing partition, then Theorem~5 of \citet{bernsteinmanorshimkin} implies that the average vector payoff approaches the corresponding closed convex image. In particular, the sender's average payoff approaches
\[
R^+(Q):=\bigcap_{\eta>0}
\mathrm{Co}\Big\{\hat u_S^\epsilon(q):d(q,Q)\leq\eta\Big\}.
\]
In words, $R^+(Q)$ is the closed convex image of the approximate indirect utility restricted to $Q$. The fully adversarial case corresponds to $Q=\Delta(\Omega)$.

By assumption, the order-$k$ Markov chain is such that the augmented chain on
$\Omega^k$ is irreducible and aperiodic, and therefore admits a unique invariant
distribution. Let $p\in\Delta(\Omega)$ denote its invariant marginal distribution
over states.

Fix $\eta>0$. Concentration inequalities for Markov chains imply that there exist constants $C_\eta,c_\eta>0$ such that, for every block length $l$,
\[
\mathbb P\big(\|\overline{\omega}_l-p\|>\eta\big)\leq C_\eta e^{-c_\eta l},
\]
see, e.g., \citet{lezaud1998} or \citet{paulin2015}.

Choose a partition into consecutive blocks with lengths $l_m=\lceil m^\nu\rceil$ for some $\nu\in(0,1)$.  By the Borel--Cantelli lemma, for every $\eta>0$, almost surely only finitely many blocks satisfy $\|\overline{\omega}_m-p\|>\eta$. Hence, $\overline{\omega}_m\to p$ almost surely, and nature's play is empirically $\{p\}$-restricted with respect to the partition $(l_m)_{m\geq1}$.  For sufficiently small $\nu$, these block lengths satisfy the growth condition of Theorem~5 of \citet{bernsteinmanorshimkin}. Consequently, the distance between the sender's average payoff
and $R^+(\{p\})$ converges to zero.

 If $p$ lies in the interior of a grid cell, then $R^+(\{p\})=\{\hat u_S^\epsilon(p)\}$. If $p$ lies on the boundary of two or more grid cells, choose $f^*(p)$ among the associated forecasts to minimize the sender's expected payoff at $p$. Since $f^*$ is constant within each adjacent cell, every limit of $\hat u_S^\epsilon(q)$ as $q\to p$ is at least $\hat u_S^\epsilon(p)$. Hence, the sender attains $\hat u_S^\epsilon(p)$. The same strategy is also $\epsilon$-calibrated, since its average calibration cost approaches $\mathcal C$, as in the proof of Theorem \ref{maxattainable}.

\subsection{Proof of Proposition~\ref{calibnimpliesnoreg}}

Fix any calibrated forecasting strategy and suppose the receiver best responds
to every forecast, so that $a_t=\hat a(f_t)$ for all $t$.
Thus, whenever $N_T[f]>0$, we have
$\overline u_{R,T}[f]=u_R(\overline{\omega}_T[f],\hat a(f))$.

Let $L_R:=\max_{a\in A}\|u_R(\cdot,a)\|$.
Since $\hat a(f)$ is a best response to $f$, for every $a\in A$,
\[
u_R(\overline{\omega}_T[f],a)
-u_R(\overline{\omega}_T[f],\hat a(f))
\leq
2L_R\|\overline{\omega}_T[f]-f\|.
\]
Therefore,
\[
\begin{aligned}
&\limsup_{T\to\infty}
\frac{N_T[f]}{T}
\left(
\max_{a\in A}u_R(\overline{\omega}_T[f],a)
-u_R(\overline{\omega}_T[f],\hat a(f))
\right)\\
&\qquad\leq
2L_R\limsup_{T\to\infty}
\frac{N_T[f]}{T}
\|\overline{\omega}_T[f]-f\|
=0,
\end{aligned}
\]
where the equality follows from calibration. Hence, the receiver has no regret
with respect to $f$.

\subsection{Proof of Proposition~\ref{noregretQ}}

Let $Q^\star\in\mathcal M(C_\mu)$ be an optimal
distribution with finite support, and let $\pi^\star$
be a signaling policy that implements it.

By the nondegeneracy condition, for each
$f\in\mathrm{Supp}(Q^\star)$ there exists a belief
$q_f\in\mathrm{Co}(\mathrm{Supp}(C_\mu))$ at which
$\hat a(f)$ is the unique best response.

We construct an auxiliary signaling policy $\rho$
that assigns each $f\in\operatorname{Supp}(Q^\star)$
positive probability $\beta_f$ and induces posterior mean $q_f$.
Write
$q_f=\sum_{p\in\operatorname{Supp}(C_\mu)}\lambda_f(p)p$,
where $\lambda_f\in\Delta(\operatorname{Supp}(C_\mu))$.
Choose $\beta_f>0$ sufficiently small that
\[
\sum_{f\in\operatorname{Supp}(Q^\star)}
\beta_f\lambda_f(p)\leq C_\mu(p)
\qquad\text{for every }p\in\operatorname{Supp}(C_\mu).
\]
For each such $f$ and $p$, set
$\rho(f\mid p)=\beta_f\lambda_f(p)/C_\mu(p)$
and assign the remaining probability to a forecast
$g\notin\operatorname{Supp}(Q^\star)$.

Fix $\eta\in(0,1)$ and let the sender use
\[
\pi^\eta=(1-\eta)\pi^\star+\eta\rho
\]
in every period. By Lemma~\ref{stationaryergodic}
and the strong law for bounded martingale differences,
for every $f\in\mathrm{Supp}(Q^\star)$ and
$a\ne\hat a(f)$, almost surely,
\[
\begin{aligned}
\frac1t
\left(U_{\hat a(f),t}^f-U_{a,t}^f\right)
\to{}&
(1-\eta)Q^\star(f)
\bigl(u_R(f,\hat a(f))-u_R(f,a)\bigr)\\
&+\eta\beta_f
\bigl(u_R(q_f,\hat a(f))-u_R(q_f,a)\bigr)>0.
\end{aligned}
\]
The first term is nonnegative because $\hat a(f)$
is a best response to $f$, and the second is
strictly positive by construction.

Since $\mathrm{Supp}(Q^\star)$ and $A$ are finite,
there exists $d_\eta>0$ such that all these
cumulative payoff differences exceed $d_\eta t$
for all sufficiently large $t$, almost surely.
In the $T$-period game, the mean-based threshold
is $\gamma_T T$, where $\gamma_T=o(1)$.
Thus, the fraction of periods in which these
differences fail to exceed $\gamma_T T$
converges to zero almost surely. 

For any fixed mean-based receiver, bounded convergence
and the mean-based property imply that the receiver
chooses $\hat a(f)$ in all but $o(T)$ periods in
expectation following forecasts in
$\mathrm{Supp}(Q^\star)$.

Under $\pi^\eta$, the sender uses $\pi^\star$
with probability $1-\eta$. By stationarity and
ergodicity, these periods would contribute
$(1-\eta)\mathrm{Per}(C_\mu)$ to the limiting
expected average payoff if the receiver always
chose $\hat a(f)$. Since payoffs are bounded,
the vanishing frequency of other actions does
not affect this limit. The auxiliary rule is
used with probability $\eta$, and its payoff
is bounded below by $-L$. Therefore,
\[
\liminf_{T\to\infty}
\frac1T\mathbb E\left[
\sum_{t=1}^T u_S(\omega_t,a_t)
\right]
\geq
(1-\eta)\mathrm{Per}(C_\mu)-\eta L
\geq
\mathrm{Per}(C_\mu)-2L\eta.
\]

Finally,  let $\eta_n=1/n$ for $n\geq2$. For each $n$, choose $T_n$ sufficiently large that using $\pi^{\eta_n}$ throughout the $T$-period game guarantees an expected average payoff of at least $\mathrm{Per}(C_\mu)-(2L+1)/n$ whenever $T\geq T_n$, with $T_n$ strictly increasing. For each horizon $T_n\leq T<T_{n+1}$, the sender uses $\pi^{\eta_n}$ in every period. As $T\to\infty$, the chosen $\eta_n$ converges to zero,
so the sender guarantees $\mathrm{Per}(C_\mu)$.

\end{document}